\documentclass[11pt]{article}

\usepackage[final]{acl}

\usepackage{times}
\usepackage{latexsym}
\usepackage{amsmath}
\usepackage{cleveref}
\usepackage{amssymb}
\usepackage{amsthm}
\usepackage{subcaption}

\usepackage[T1]{fontenc}

\usepackage[utf8]{inputenc}

\usepackage{microtype}

\usepackage{inconsolata}

\usepackage{graphicx}

\usepackage[ruled,vlined]{algorithm2e}
\usepackage{bm}
\usepackage{longtable}
\usepackage{booktabs}
\usepackage{subcaption}
\usepackage{dsfont}
\usepackage{paralist}

\newtheorem{assumption}{Assumption}

\newtheorem{lemma}{Lemma}

\Crefname{section}{Sec.}{Secs.}
\Crefname{figure}{Fig.}{Figs.}
\Crefname{equation}{Eq.}{Eqs.}
\Crefname{table}{Tab.}{Tabs.}

\usepackage{xcolor,colortbl}
\usepackage{pifont}         
\usepackage{color, soul}
\usepackage{circledtext}
\usepackage{circledsteps}
\definecolor{myorange}{RGB}{230, 159, 0}
\definecolor{mygreen}{RGB}{0, 158, 115}
\definecolor{myred}{RGB}{222,48,47}
\definecolor{myblue}{RGB}{52,89,156}

\circledtextset{resize=real}

\usepackage{enumitem}
\usepackage{adjustbox}
\usepackage{listings}

\newcommand{\hfdataset}{
\href{https://huggingface.co/datasets/sumin-yu/PopResume}{
\raisebox{-0.2em}{\includegraphics[height=1.2em]{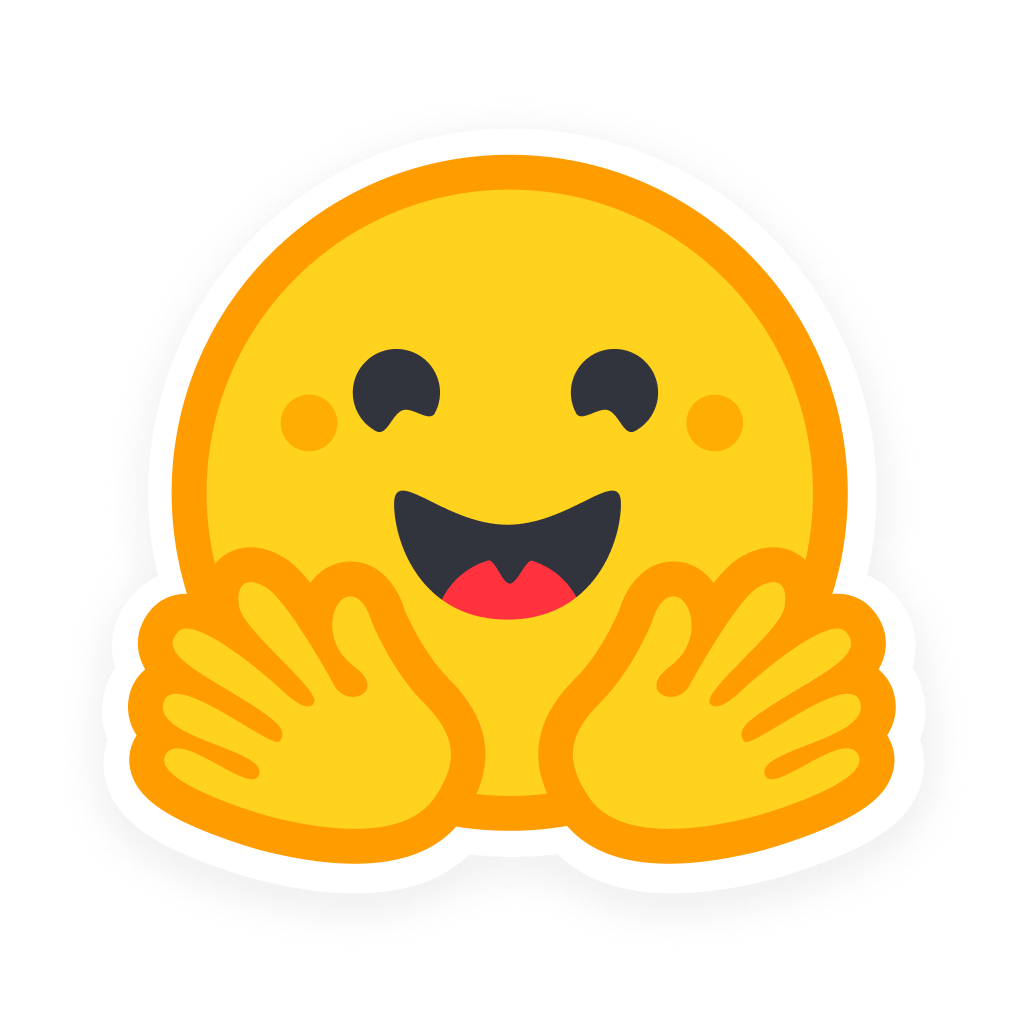}}
\;\texttt{sumin-yu/PopResume}
}}

\newcommand{\githubrepo}{
\href{https://github.com/sumin-yu/PopResume}{
\raisebox{-0.18em}{
\includegraphics[height=1.0em]{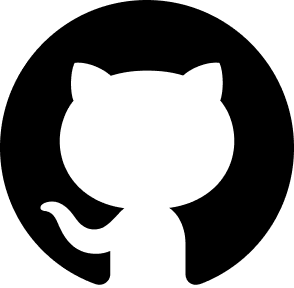}}
\;\texttt{sumin-yu/PopResume}
}}
\title{
PopResume: Causal Fairness Evaluation of LLM/VLM Resume Screeners\\with Population-Representative Dataset}

\author{Sumin Yu$^{1}$\thanks{Equal contribution.}, Juhyeon Park$^{2}$\footnotemark[1], Taesup Moon$^{1,2,3}$\thanks{Corresponding author.} \\
$^{1}$ECE, Seoul National University,
$^{2}$IPAI, Seoul National University, \\ 
$^{3}$ASRI / INMC / AIIS, Seoul National University\\
\texttt{\{ysmsoomin, parkjh9229, tsmoon\}@snu.ac.kr} \\
}

\begin{document}
\maketitle
\begin{abstract}
We present PopResume, a population-representative resume dataset for causal fairness auditing of LLM- and VLM-based resume screening systems. Unlike existing benchmarks that rely on manually injected demographic information and outcome-level disparities, PopResume is grounded in population statistics and preserves natural attribute relationships, enabling path-specific effect (PSE)-based fairness evaluation.
We decompose the effect of a protected attribute on resume scores into two paths: the business necessity path, mediated by job-relevant qualifications, and the redlining path, mediated by demographic proxies. This distinction allows auditors to separate legally permissible from impermissible sources of disparity. Evaluating four LLMs and four VLMs on PopResume's 60.8K resumes across five occupations, we identify five representative discrimination patterns that aggregate metrics fail to capture. Our results demonstrate that PSE-based evaluation reveals fairness issues masked by outcome-level measures, underscoring the need for causally-grounded auditing frameworks in AI-assisted hiring. Our code is available at \githubrepo and dataset is available at \hfdataset.
\end{abstract}

\section{Introduction}
\label{sec:intro}
The use of large language models (LLMs) in resume screening and hiring decisions is expanding rapidly~\cite{lo2025ai,gan2024application}, raising growing concerns regarding fairness and discrimination~\cite{cherepanova2025improving,wilson2024gender,wilson2025no}.
Given the potential for automated decisions to produce systematic disparities that disadvantage individuals based on protected attributes such as gender or race,
organizations deploying automated hiring systems are now expected to demonstrate compliance with anti-discrimination law~\cite{TitleVII1964,EEOC_UGESP_QA}.

\begin{figure}[t]
\centering
\includegraphics[width=\linewidth]{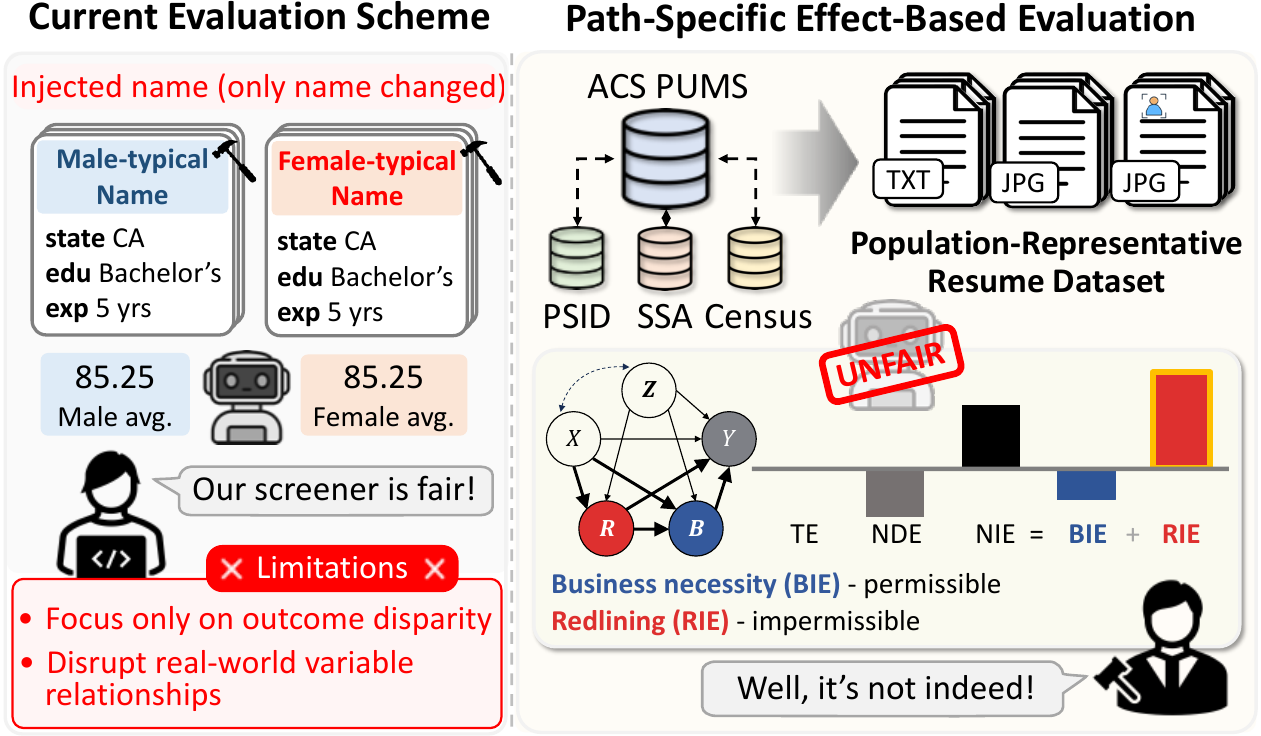}

\caption{
Prior works inject protected attributes, making causal framework–based evaluation infeasible, and measure outcome disparity; our population-representative resume dataset instead enables causal effect–based evaluation and pathway decomposition.
}
\label{fig:teaser}
\end{figure}

In actual audit and compliance environments, group-level outcome disparities are often identified using simple statistical criteria such as selection rate comparisons~\cite{EEOC_UGESP_QA}.
While these metrics serve as useful \textit{early} indicators of potential risk,
they \textit{do not explain the mechanisms behind the disparities}. This limitation is critical, as legal liability ultimately depends not only on the presence of disparity, but also on how it arises.
Recent studies investigating bias in LLM-based hiring systems \cite{iso2025evaluating,hu2025fairwork, wang2024jobfair} share the same limitation: they primarily focus on statistical disparities between demographic groups, without identifying the underlying causal pathways through which protected attributes influence hiring outcomes, as illustrated in \Cref{fig:teaser} (left).
However, legal and policy frameworks around employment discrimination hinge on a crucial distinction: disparities arising through legitimate job-related pathways, \textit{e.g.,} role-essential qualifications, may be permissible (\textit{business necessity}), whereas disparities driven by proxy variables such as name or address that can disadvantage specific protected groups are unacceptable (\textit{redlining}).
Therefore, meaningful fairness auditing requires understanding not just the magnitude of outcome disparities, but the causal pathways through which protected attributes produce them.

Path-specific effect (PSE)-based fairness frameworks~\cite{plecko2024causal,pearl2009causality, chiappa2019path} offer a principled approach to this challenge, enabling answering \textit{how} and \textit{through which pathways} protected attributes contribute to observed disparities.
However, enabling such pathway-level analysis requires data with explicit control over all relevant attributes, such as protected attributes, legitimate qualifications, and proxy variables, simultaneously.
Existing benchmarks for LLM-based hiring systems fail to meet this requirement in two distinct ways.
Real-world resume datasets~\cite{job-recommendation, senger2025karrierewege} typically do not contain ground-truth demographic information about applicants, as such attributes are withheld due to privacy concerns.
Current evaluation benchmarks~\cite{hu2025fairwork, wang2024jobfair, iso2025evaluating, armstrong2024silicon, wilson2024gender, nghiem2024you}, on the other hand, often manually inject and perturb demographic information, which can disrupt the realistic relationships between attributes, making causal framework-based evaluation infeasible.

We address these gaps by introducing legally-grounded causal fairness auditing framework for LLM- and vision language model (VLM)-based resume screening systems.
Building on the PSE framework, we decompose the causal effect of a protected attribute on resume scores into two pathways: (1) the \textit{business necessity} path, mediated by job-relevant qualifications (education, work experience), and (2) the \textit{redlining} path, mediated by demographic proxies (name, address). This decomposition enables auditors to identify not only that a hiring system discriminates, but how and whether the discrimination is legally defensible.
To support this analysis under realistic conditions,
we construct \textbf{PopResume}, a dataset derived from U.S. population-level statistics (\textit{e.g.,} ACS PUMS~\cite{acs_pums}) that preserves the natural relationships among attributes and enables causal effect estimation, which contains 60.8k resumes across five occupations.

Our contributions are as follows.
\begin{itemize}[nosep,leftmargin=1em,labelwidth=*,align=left]
  \sloppy
    \item We propose a PSE-based fairness auditing framework for LLM/VLM-based resume screening systems that separates legally permissible outcome disparities mediated by job-relevant qualifications from legally impermissible ones transmitted through demographic proxies.
    \item We construct a large-scale resume dataset grounded in U.S. population statistics,
    preserving realistic distribution and
    enabling causal effect-based fairness evaluation.
\end{itemize}

\section{Preliminaries}
\label{sec:prelim}
First, we review the legal foundations of employment discrimination law that motivate our fairness criteria (\Cref{subsec:legal-background}), and then introduce the causal modeling framework and PSEs that underpin our analysis (\Cref{subsec:sfm}).

\subsection{Legal Background}
\label{subsec:legal-background}

Title VII of the US Civil Rights Act~\cite{TitleVII1964} prohibits employment discrimination based on race, color, religion, sex, and national origin, which are referred to as \textit{protected attributes}.
The legal theory typically applied to automated decisions is the \textit{disparate impact} doctrine~\cite{jones2022disparateimpact,barocas2016disparateimpact},
which holds that even a seemingly neutral practice can be unlawful if it produces outcomes that disproportionately disadvantage a protected group.
To establish disparate impact, plaintiffs must first demonstrate a disparity in outcomes between groups, commonly assessed using statistical heuristics such as the Four-Fifths rule.
This step identifies the \textit{presence} of a disparity but does not determine its causal origin.
Critically, however, disparate impact analysis does not end with identifying statistical disparities.
Under the burden-shifting framework established by Title VII,
employers may justify the observed disparity by demonstrating that it arises from job-related criteria consistent with \textit{business necessity}, recognizing that disparities mediated through legitimate, job-relevant factors, such as professional qualifications or required skills, are permissible.
Accordingly, legal liability hinges on a fundamentally causal question:
whether the observed disparity arises through permissible pathways (\textit{business necessity}) or impermissible ones (\textit{redlining}).
This motivates the PSE-based analysis that can distinguish between these two types of effects.
 
\begin{figure}[t]
\centering

\begin{subfigure}{0.45\linewidth}
    \centering
    \includegraphics[width=\linewidth]{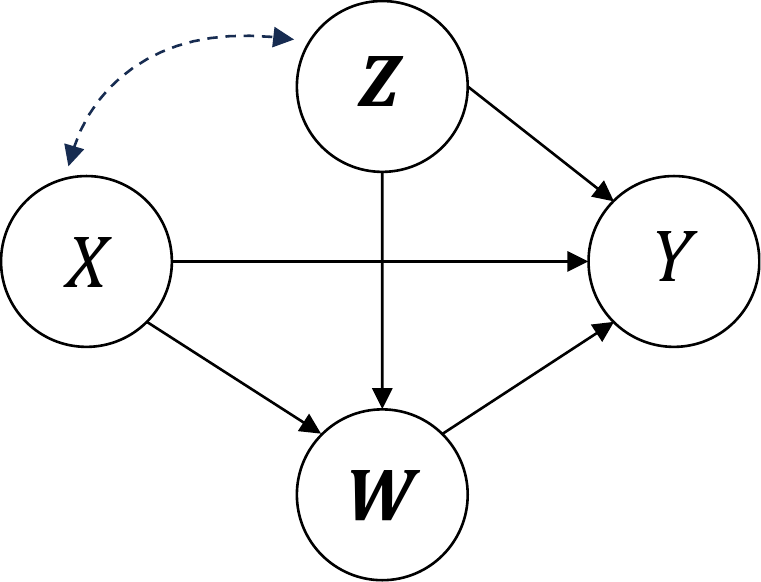}
    \caption{Standard Fairness Model (SFM).}
    \label{fig:causal_diagram}
\end{subfigure}
\hfill
\begin{subfigure}{0.45\linewidth}
    \centering
    \includegraphics[width=\linewidth]{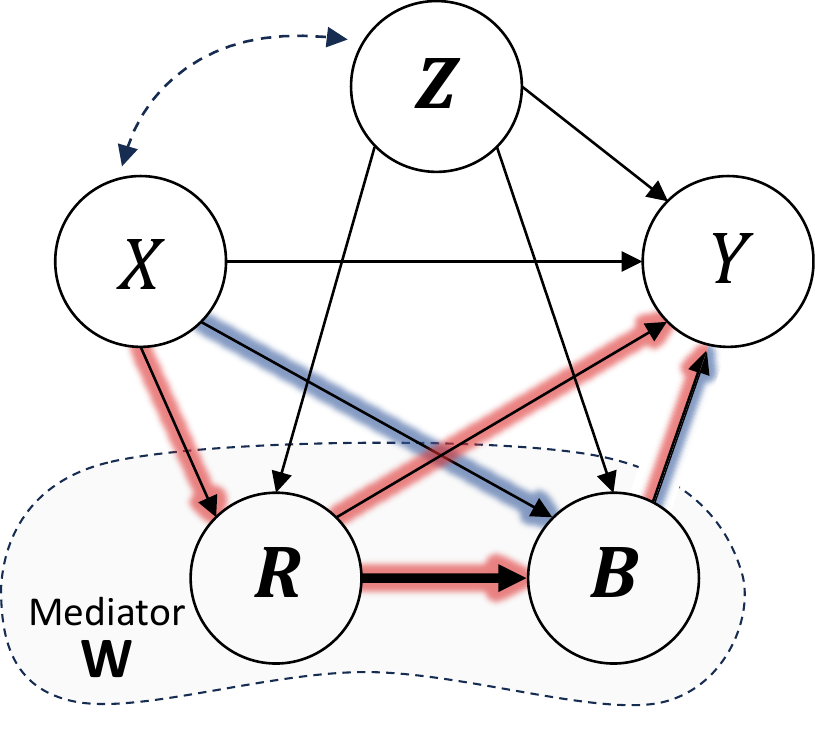}
    \caption{Mediator decomposition in our framework.}
    \label{fig:causal_diagram_ours}
\end{subfigure}
\caption{
In our framework, the mediator set $\mathbf{W}$ is decomposed into business necessity components $\mathbf{B}$ and redlining-related components $\mathbf{R}$.
{
Paths from $X$ to $Y$ through $\mathbf{B}$ (and not through $\mathbf{R}$) constitute the \textcolor{myblue}{BIE}, while all paths passing through $\mathbf{R}$ constitute the \textcolor{myred}{RIE}.
}
}
\end{figure}

\subsection{Standard Fairness Model \& Causal Effects}
\label{subsec:sfm}

We use notations $(\mathbf{X}, X, \mathbf{x},x)$ to denote random vectors, random variables, and corresponding realizations, respectively.

\noindent \textbf{Standard Fairness Model (SFM).}
To formalize the causal relationships between protected attributes and hiring outcomes, 
we build on the structural causal model (SCM) framework~\cite{pearl2009causality}.
An SCM induces a causal graph $\mathcal{G}$, where nodes correspond to variables $\mathbf{V}$, directed edges represent causal relationships, and dashed bi-directed edges indicate hidden confounders.

Following ~\citet{plecko2024causal}, we adopt the standard fairness model (SFM) with $\mathbf{V}=\{X,\mathbf{Z},\mathbf{W},Y\}$.
SFM partitions variables into four groups: a protected attribute $X$, an outcome $Y$, confounders $\mathbf{Z}$ that are not causally affected by $X$, and mediators $\mathbf{W}$ that may be causally influenced by $X$, as depicted in \Cref{fig:causal_diagram}. This abstraction allows fairness analysis without requiring the complete causal graph to be specified, while preserving the causal pathways through which the protected attribute $X$ may influence the outcome $Y$.

\noindent \textbf{Path-specific Effects (PSE).}
Under the SFM, the influence of the protected attribute ${X}$ on the outcome ${Y}$ can propagate through multiple pathways.
To quantify these effects, we adopt the potential outcome ~\cite{pearl2009causality}.
Given variables ${X}, {Y} \subseteq \mathbf{V}$, the \textit{average potential outcome} $\mathbb{E}[{Y}_{{x}}]$ denotes the expected value of ${Y}$ if $\mathbf{X}$ were set to ${x}$ by intervention.
\textit{Total effect} (TE) of changing $X$ from $x_0$ to $x_1$ is defined as
\begin{equation}
    \text{TE}(x_1, x_0) \triangleq \mathbb{E}[Y_{x_1}] - \mathbb{E}[Y_{x_0}], 
    \label{eq:te}
\end{equation}
capturing the overall causal effect of $X$ on $Y$ through all causal pathways.
TE can be decomposed into the \textit{natural direct effect} (NDE) and the \textit{natural indirect effect} (NIE).
NDE captures the portion of the effect of $X$ on $Y$ that is \textit{not} mediated through the mediators $\mathbf{W}$,
\begin{equation}
    \text{NDE}(x_1, x_0) \triangleq \mathbb{E}[Y_{x_1, \mathbf{W}_{x_0}}] - \mathbb{E}[Y_{x_0}],
    \label{eq:nde}
\end{equation}
while NIE captures the effect transmitted through the mediators $\mathbf{W}$,
\begin{equation}
    \text{NIE}(x_1, x_0) \triangleq \mathbb{E}[Y_{x_1}] - \mathbb{E}[Y_{x_1, \mathbf{W}_{x_0}}].
    \label{eq:nie}
\end{equation}

\section{PSE-based Fairness Evaluation}
\label{sec:pse_based_fair_eval}

\subsection{Mediator Decomposition}
As discussed in \Cref{subsec:legal-background}, legal doctrines such as \textit{redlining} and \textit{business necessity} require distinguishing the causal pathways through which protected attributes influence hiring outcomes.
To capture this distinction within the causal framework, we decompose the mediator set $\mathbf{W}$ into two components: $\mathbf{B}$, capturing job-relevant qualifications associated with \textit{business necessity}, and $\mathbf{R}$, capturing demographic proxies associated with \textit{redlining},
as illustrated in~\Cref{fig:causal_diagram_ours}\footnote{While this decomposition is inspired by legal doctrines in employment discrimination, the normative interpretation of these pathways depends on legal and societal contexts.}.
Based on this decomposition, we introduce two PSEs, termed  the \textit{$\mathbf{B}$-specific indirect effect} (BIE), and the \textit{$\mathbf{R}$-specific indirect effect} (RIE).
Intuitively, BIE captures the portion of the causal effect transmitted through $\mathbf{B}$, while RIE captures the portion transmitted through $\mathbf{R}$, corresponding to the pathways highlighted by the blue and red edges in \Cref{fig:causal_diagram_ours}.
Because $\mathbf{B}$ and $\mathbf{R}$ may causally influence each other, the \textit{order of interventions} on these mediators is not uniquely defined.
Specifically,
when considering BIE,
two intervention orders are possible:
(i) fixing $\mathbf{R}$ at $\mathbf{R}_{x_0}$ and intervening on $\mathbf{B}$, or (ii) first intervening on $\mathbf{R}$ from $\mathbf{R}_{x_0}$ to $\mathbf{R}_{x_1}$ and then intervening on $\mathbf{B}$.
Therefore, we adopt a symmetric formulation that averages over the two possible mediator intervention orders.
Formally, BIE is defined as
\begin{equation}
\begin{aligned}
\text{BIE}&(x_1,x_0)
\triangleq
\frac{1}{2}\Big[
\underbrace{
\mathbb{E}[Y_{x_1}]-\mathbb{E}[Y_{x_1,\mathbf{B}_{x_0,\mathbf{R}_{x_1}}, \mathbf{R}_{x_1}}]}_{\text{\scriptsize (ii)}} \\
& +
\underbrace{
\mathbb{E}[Y_{x_1,\mathbf{B}_{x_1,\mathbf{R}_{x_0}}, \mathbf{R}_{x_0}}] - \mathbb{E}[Y_{x_1, \mathbf{B}_{x_0}, \mathbf{R}_{x_0}}]
}_{\text{\scriptsize (i)}}
\Big].
\label{eq:bie}
\end{aligned}    
\end{equation}

Similarly, RIE is defined as
\begin{equation}
\begin{aligned}
\text{RIE}&(x_1,x_0)
\triangleq 
\frac{1}{2}\Big[
\mathbb{E}[Y_{x_1}]
-
\mathbb{E}[Y_{x_1,\mathbf{B}_{x_1,\mathbf{R}_{x_0}}, \mathbf{R}_{x_0}}]
 \\
& +
\mathbb{E}[Y_{x_1,\mathbf{B}_{x_0,\mathbf{R}_{x_1}}, \mathbf{R}_{x_1}}]
-
\mathbb{E}[Y_{x_1,\mathbf{B}_{x_0}, \mathbf{R}_{x_0}}]
\Big].
\label{eq:rie}
\end{aligned}
\end{equation}

By construction, NIE is decomposed into the sum of BIE and RIE, \textit{i.e.,} $\text{NIE}(x_1,x_0) = \text{BIE}(x_1,x_0) + \text{RIE}(x_1,x_0)$.

\subsection{PSE Estimation}
\label{sec:causal_effect_estimation}
\begin{table}[t]
\centering
\small
\resizebox{\columnwidth}{!}{
\begin{tabular}{cll}
\toprule
& $\mu^i$& $\pi^i$ \\
\midrule
$i=3$ & $\mathbb{E}[Y|\mathbf{B,R},X,\mathbf{Z}]$ & $\frac{P(X=x_0|\mathbf{R,Z})}{P(X|\mathbf{R,Z})} \frac{\mathds{1}[X={x_1}]}{P(X=x_0|\mathbf{Z})}$ \\
$i=2$ & $\mathbb{E}[\mu^3(\mathbf{B,R},x_1,\mathbf{Z})|\mathbf{R},X,\mathbf{Z}]$ & $\frac{P(X=x_0|\mathbf{R,Z})}{P(X|\mathbf{R,Z})} \frac{\mathds{1}[X={x_1}]}{P(X=x_0|\mathbf{Z})}$ \\
$i=1$ & $\mathbb{E}[\mu^2(\mathbf{R},x_1,\mathbf{Z})|X,\mathbf{Z}]$ & $\frac{\mathds{1}[X={x_0}]}{P(X|\mathbf{Z})}$ \\
\bottomrule
\end{tabular}
}
\caption{Definition of nuisance parameters.}
\label{tab:nuisance_parameters}
\end{table}
\noindent \textbf{Estimation of PSE.}
To estimate the PSEs defined in~\Cref{eq:te,eq:nde,eq:nie,eq:rie,eq:bie},
we need to estimate five quantities from the given data,
$\mathbb{E}[Y_{x_0}],\  \mathbb{E}[Y_{x_1}],\  \mathbb{E}[Y_{x_1, \mathbf{W}_{x_0}}], \ \mathbb{E}[Y_{x_1,\mathbf{B}_{x_0,\mathbf{R}_{x_1}}, \mathbf{R}_{x_1}}]$, and $\mathbb{E}[Y_{x_1,\mathbf{B}_{x_1,\mathbf{R}_{x_0}}, \mathbf{R}_{x_0}}]$.
The estimation of the first four quantities 
are already discussed by \citet{zhang2025path}. Therefore, we detail the estimation of $\mathbb{E}[Y_{x_1,\mathbf{B}_{x_1,\mathbf{R}_{x_0}}, \mathbf{R}_{x_0}}]$, firstly introduced in our work. Following the DML-UCA algorithm of \citet{NEURIPS2024_0c4bc137}, a doubly robust estimator with finite sample guarantee can be constructed as follows.

\begin{lemma}[DML-UCA, \cite{NEURIPS2024_0c4bc137}]
Given a sample $\mathcal{D} \overset{i.i.d.}{\sim} P(\mathbf{V})$, doubly robust estimator $\hat{\psi}$ for $\mathbb{E}[Y_{x_1,\mathbf{B}_{x_1,\mathbf{R}_{x_0}}, \mathbf{R}_{x_0}}]$ constructed with following procedure has finite sample guarantee:
\begin{enumerate}[nosep,leftmargin=1em,labelwidth=*,align=left]
\item
Take any $L$-fold random partition of the dataset 
$\mathcal{D} \triangleq (\mathbf{V}_1, \ldots, \mathbf{V}_n)$,
$\mathcal{D} = \bigcup_{\ell=1}^{L} \mathcal{D}_\ell$
where $|\mathcal{D}_\ell| = n/L$.
\item
For each $\ell = 1,2,\ldots,L$, construct nuisance parameter estimators 
$\hat{\bm{\mu}}=(\hat{\mu^1},\hat{\mu^2},\hat{\mu^3})$ and 
$\hat{\bm{\pi}}=(\hat{\pi^1},\hat{\pi^2},\hat{\pi^3})$ using 
$\mathcal{D} \setminus \mathcal{D}_\ell$, and compute
$\hat{\psi}_\ell \triangleq \mathbb{E}_{\mathcal{D}_\ell} \varphi(\mathbf{V};\hat{\bm{\mu}}, \hat{\bm{\pi}}),$
where $\bm{\mu}, \bm{\pi}$ are outlined in \Cref{tab:nuisance_parameters} and
\begin{equation*}
\begin{aligned}
&\varphi(\mathbf{V};\hat{\bm{\mu}},\hat{\bm{\pi}})
\triangleq
\hat{\pi}^3\{Y-\hat{\mu}^3(\mathbf{B,R},X,\mathbf{Z})\} \\
&+ \hat{\pi}^2\{\hat{\mu}^3(\mathbf{B,R},x_1,\mathbf{Z})-\hat{\mu}^2(\mathbf{R},X,\mathbf{Z})\} \\
&+ \hat{\pi}^1\{\hat{\mu}^2(\mathbf{R},x_1,\mathbf{Z})-\hat{\mu}^1(X,\mathbf{Z})\} + \hat{\mu}^1(x_0,\mathbf{Z}).
\end{aligned}
\end{equation*}

\item
The estimator $\hat{\psi}$ is the average of 
$\{\hat{\psi}_\ell\}_{\ell=1}^{L}$,
$
\hat{\psi} 
\triangleq 
\frac{1}{L}(\hat{\psi}_1 + \cdots + \hat{\psi}_L).
$
\end{enumerate}
\end{lemma}
Mathematical details are provided in \Cref{sec:app_identi,sec:app_paramet,sec:app_doubly_robust}, and the finite sample guarantee is inherited from Theorem~2 in \citet{zhang2025path}. We estimate the nuisance parameters using XGBoost \cite{chen2016xgboost}; corresponding hyperparameter search details are provided in \Cref{sec:hparam}.

\noindent \textbf{Grouping of Categorical Variables.}
The estimation process 
requires
estimating $\mathbb{E}[Y|\cdot]$ and $P(X|\cdot)$.
However, in practice, some discrete variables may have a large number of categories, which can lead to unstable estimation when rare categories are unevenly represented across partitions.
To address this issue, we group the categories of such variables into super-categories, reducing the number of categories and enabling more reliable estimation. The specific variables and grouping strategies used in our dataset are described in \Cref{sec:app_grouping}.

\section{PopResume}
\label{sec:dataset-construction}

\begin{figure*}[t]
\centering
\includegraphics[width=\textwidth]{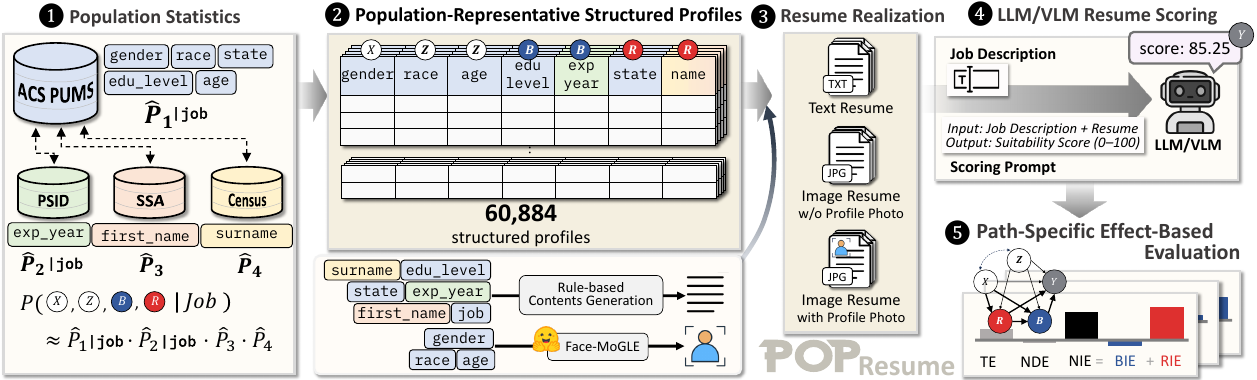}
\caption{
Pipeline for constructing the population-representative resume dataset and evaluating LLM/VLM-based resume screeners.
\circledtext*[height=1.9ex]{1}
Estimation of joint distribution $P(X, \mathbf{Z,B,R}|J)$ based on Assumption 1.
\circledtext*[height=1.9ex]{2}
Population-representative structured profiles consisting of 
protected attribute
\Circled[inner color=black, fill color=white, outer color=none, inner xsep=4.5pt, inner ysep=4.5pt]{X},
confounder
\Circled[inner color=black, fill color=white, outer color=none, inner xsep=4.5pt, inner ysep=4.5pt]{Z},
business necessity mediators \Circled[inner color=white, fill color=myblue, outer color=none, inner xsep=4.5pt, inner ysep=4.5pt]{B}, and
redlining mediators 
\Circled[inner color=white, fill color=myred, outer color=none, inner xsep=4.5pt, inner ysep=4.5pt]{R}.
\circledtext*[height=1.9ex]{3}
Resume realization, where each structured profile is converted into a natural-language resume using rule-based procedures conditioned solely on assigned attributes, eliminating uncontrolled variation. Three formats are produced: text resumes for LLM evaluation, and resume images with and without synthesized profile photos for VLM evaluation.
\circledtext*[height=1.9ex]{4}
Resume scoring by LLM/VLM screeners, which assign a score Y given a job description and a resume.
\circledtext*[height=1.9ex]{5}
Path-specific effect-based evaluation, estimating TE, NDE, and NIE, and further decomposing NIE into 
{\color{myblue}{BIE}} and {\color{myred}{RIE}}.
}
\label{fig:popresume_main}
\end{figure*}
Estimating the PSEs introduced in \Cref{sec:pse_based_fair_eval} requires access to all variables in $\mathbf{V}$, including the protected attribute $X$, confounders $\mathbf{Z}$, and mediators $\mathbf{B}$ and $\mathbf{R}$.
However, as discussed in \Cref{sec:intro}, publicly available resume datasets typically lack annotations for protected attributes and other relevant covariates.
To address this limitation, we construct \textbf{PopResume}, a population-representative resume dataset grounded in U.S. population statistics.
The dataset approximates the natural joint distribution of $\mathbf{V}\setminus\{Y\}$, preserving realistic relationships among the variables.

\Cref{fig:popresume_main} illustrates the overall dataset construction pipeline, which consists of two stages.
First, we generate population-representative structured profiles from multiple statistical sources
(\circledtext*[height=1.9ex]{1}-\circledtext*[height=1.9ex]{2} 
in \Cref{fig:popresume_main}, \Cref{subsec:data_construction_1_2}). Second, each profile is converted into a realistic natural-language resume through rule-based instantiation
(\circledtext*[height=1.9ex]{3} 
in \Cref{fig:popresume_main}, \Cref{subsec:data_construction_3}).

\begin{table}[t]
\centering
\resizebox{0.85\linewidth}{!}{%
\begin{tabular}{llll}
\toprule
 & Resume Attribute & Operational Variable & Data Source \\
\midrule
$X,\mathbf{Z}$ & Gender, Race, Age & \textit{gender, race, age} & ACS PUMS \\
$\mathbf{B}$ & Work experience & \textit{exp\_year} & PSID \\
$\mathbf{B}$ & Education & \textit{edu\_level} & ACS PUMS \\
$\mathbf{R}$ & Name & \textit{first\_name}, \textit{surname} & SSA, Census \\
$\mathbf{R}$ & Address & \textit{state} & ACS PUMS \\
\bottomrule
\end{tabular}%
}
\caption{Mapping between conceptual variables and their operational representations.
}
\label{tab:resume_attribute_mapping}
\end{table}
Before introducing our dataset construction pipeline in detail, we clarify how resume attributes are operationalized in our structured representation.
\Cref{tab:resume_attribute_mapping} summarizes the mapping between conceptual variables, resume attributes, and their operational variables in dataset construction.
While the mapping of $X$ and $\mathbf{Z}$ is straightforward,
the mediator variables $\mathbf{B}$ and $\mathbf{R}$ do not appear directly in population statistics, although they conceptually correspond to resume components such as work experience, education, name, and address.
Instead, we represent them using observable variables that can be reliably obtained from public datasets.
Specifically, name is represented using sampled \textit{first\_name} and \textit{surname}, and address is represented at the state level (\textit{state}).
Work experience is represented as \textit{years of experience} (\textit{exp\_year}), and education as a categorical \textit{education level} (\textit{edu\_level}).
These operational variables allow us to represent both job-related qualifications and demographic proxy signals within a unified structured profile.
We note that the categorization of variables into $\mathbf{B}$ and $\mathbf{R}$ reflects one plausible operationalization; alternative categorizations are possible depending on the application context.

\subsection{Population-Representative Dataset}
\label{subsec:data_construction_1_2}
To approximate realistic applicant distributions, we construct structured resume profiles grounded in U.S. population survey data. Our objective is to approximate the joint distribution for each job $J$, \textit{i.e.,} $P(X,\mathbf{Z},\mathbf{B},\mathbf{R} \mid J)$,  which captures demographic and job-related characteristics of applicants within each occupation.
Because available statistical sources provide only partial demographic and occupational information rather than a unified resume-level dataset, directly obtaining the full joint distribution is not feasible. 
Therefore, we factorize the joint distribution as follows. We then estimate each component using different source datasets and perform ancestral sampling based on
the ACS PUMS \cite{acs_pums}, a national survey of the U.S. population.
\begin{assumption}
\textit{First\_name} depends only on \{\textit{gender, age}\}, \textit{surname} depends only on \{\textit{race}\}, and the source datasets used to estimate the components are mutually consistent.
\label{assumption:factorization}
\end{assumption}
Based on Assumption 1, the following holds.
\begin{equation*}
\begin{aligned}
    &P(\mathbf{V}\setminus \{Y\} \mid J)=\underbrace{P(X,\mathbf{Z},\text{\textit{state}},\text{\textit{edu\_level}}\mid J)}_{P_1}\\
        & \cdot \underbrace{P(\text{\textit{exp\_year}} \mid X, \mathbf{Z},\textit{state}, \textit{edu\_level} ,J)}_{P_2} \cdot\\& \underbrace{P(\text{\textit{first\_name}} \mid \text{\textit{gender, age}})}_{P_3} \cdot \underbrace{P(\text{\textit{surname}} \mid \text{\textit{race}})}_{P_4}
\end{aligned}
\end{equation*}

\noindent \textbf{$\hat{P_1}$ from ACS PUMS.}
We first obtain samples from $P_1$ for each occupation using the 2023 ACS PUMS~\cite{acs_pums}, a large-scale population survey that provides individual-level demographic, geographic, and occupational information for the U.S. population. We then denote the empirical distribution of these samples as $\hat{P}_1$.
Specifically, we restrict the ACS PUMS population to individuals who are currently employed and possess valid occupation codes. This ensures that the sampled individuals represent plausible job applicants in the labor market.
To ensure sufficient statistical support within each occupation, we sort jobs by sample size and randomly select five occupations
from the top twenty most represented categories. For each selected occupation, we sample individuals to obtain \textit{gender}, \textit{race}, \textit{age}, \textit{state}, and \textit{edu\_level}.
We restrict the age range to $18 \le \text{\textit{age}} \le 44$ to focus on early- and mid-career workers who constitute the primary pool of active job seekers while excluding minors who are subject to legal employment restrictions.

\noindent \textbf{$\hat{P_2}$ from PSID.} In the next step, to generate realistic \textit{exp\_year} values, we leverage data from the Panel Study of Income Dynamics (PSID)~\cite{PSID}, a longitudinal household survey that tracks employment histories, income, and demographic characteristics of individuals in the United States. 
Our goal is to approximate the $P_2$.
To account for the distribution shift in marginal distribution between PSID and ACS PUMS, we adopt density ratio matching \cite{sugiyama2012density}
to reweight samples from PSID.
We train an XGBoost regressor $f(\cdot)$ \cite{chen2016xgboost} to approximate the mean of the conditional distribution under the following assumption.
\begin{assumption}\label{assumption:1}
$P_2$ follows truncated normal distribution, 
$P_2 \sim \text{TruncNormal}\big( f(\cdot), \sigma, [0,M] \big)$.
\end{assumption}
where $\text{TruncNormal}(\mu,\sigma,[0,M])$ denotes truncated normal distribution with mean $\mu$, standard deviation of $\sigma$, and support of $[0,M]$.
Notably, the $\sigma$ is estimated hierarchically based on \textit{age}, \textit{gender}, and $J$, derived from the residuals of the training samples.
In addition, $M=\text{\textit{age}}-18$.
To check the validity of this assumption, we visualized the Q-Q plot in \Cref{sec:qq_plot}, and confirm that the normality assumption is acceptable.

\noindent \textbf{$\hat{P_3}$ from SSA \& $\hat{P_4}$ from Census.} 
We obtain the empirical distribution $\hat{P_3}$ from the Social Security Administration (SSA) name records~\cite{ssa_firstname}, which provides yearly counts of first names by gender.
First names are sampled with conditioning on gender and age to reflect gender-specific naming patterns and preserve generational naming trends.
Then, $\hat{P_4}$ is obtained from the 2010 U.S. Census surname statistics~\cite{census_surnames}, which provide race-conditioned surname frequencies.
Names sampled from these distributions serve as observable demographic signals within resumes and do not introduce additional dependencies beyond the specified conditioning variables.

\subsection{Resume Generation Pipeline}
\label{subsec:data_construction_3}
After constructing the structured profiles containing all attributes $(X,\mathbf{Z},\mathbf{B},\mathbf{R})$ for each job $J$, we convert each profile into a natural-language resume.
The overall resume structure is informed by publicly available resume templates from LiveCareer\footnote{\url{https://www.livecareer.com/}}.
Resume content is generated using rule-based procedures conditioned solely on structured variables.
This design eliminates unintended variation caused from uncontrolled textual randomness, ensuring that any differences in model outputs can be attributed to $(X, \mathbf{Z, B, R})$.
For example, education entries are instantiated according to \textit{edu\_level}, with graduation year determined by age and degree type. School names, majors, and bullet descriptions are sampled from job-specific pools of fictional institutions and content.
Similarly, contact information (e.g., phone numbers and email addresses) is synthetically generated and does not introduce additional dependencies beyond the structured variables.
The detailed rules for instantiating each resume entry are provided in \Cref{sec:app_instantiation}.
In total, the dataset contains 60,884 resumes across five jobs.
Detailed dataset statistics are provided in \Cref{sec:data_stat}.

Additionally, we also consider a resume scoring scenario in which resumes are accompanied by profile photos and evaluated using a vision language model (VLM).
To support this setting, we generate synthetic profile photos conditioned on demographic attributes, \textit{i.e.,} \textit{gender}, \textit{race}, and \textit{age}.
Specifically, we use Face-MoGLE \cite{zou2025mixture} to generate images with prompts ``a professional headshot portrait photo of a \{\textit{age}\} year old \{\textit{race}\} \{\textit{gender}\}, white background, formal attire, front-facing''.
We validate whether
the generated images correctly reflect the specified attributes,
as detailed in \Cref{sec:app_image_check}.
Unlike text resumes, which are constructed by deterministic rules, profile photos are generated with a text-to-image model,
since visual attributes can admit substantial within-group variation, and using a single fixed reference image for each demographic configuration could itself introduce systematic bias by failing to capture such variability.
However, we acknowledge that this generation process may introduce unintended variations in attributes not explicitly controlled by our prompts, such as hairstyle, eyeglasses, or other appearance characteristics.
Example resumes rendered in three formats--text resumes, image resumes without profile photos, and image resumes with profile photos--are illustrated in \Cref{fig:resume_examples}.

\section{Experiments}
\subsection{Experimental Setup}

\paragraph{Models.}
We evaluate four LLMs for text-based resume scoring:
Llama-3.1-8B-Instruct~\cite{grattafiori2024llama}, Mistral-7B-Instruct-v0.2~\cite{jiang2023mistral7b}, GPT-4o-mini~\cite{hurst2024gpt-4o}, and Gemini-2.5-Flash-Lite~\cite{comanici2025gemini}.
For image-based resume evaluation,
we consider four VLMs: Qwen2.5-VL-7B-Instruct~\cite{bai2025qwen25vltechnicalreport}, InternVL2-8B~\cite{chen2024internvl}, GPT-4o~\cite{hurst2024gpt-4o}, and Gemini-2.5-Flash-Lite.

\paragraph{Protected Attributes.}
While our framework is agnostic to the choice of protected attributes, we focus on \textit{gender} in the main experiments and provide results for \textit{race} in~\Cref{sec:full_additional_results}.
In the following experiments, we set $x_0=\textit{Female}$ and $x_1=\textit{Male}$.
Namely, a positive effect value indicates that the model favors male candidates relative to female candidates through the corresponding pathway, whereas a negative value indicates the opposite.

\paragraph{Scoring.}
Each resume is evaluated alongside a job description of the target occupation, and the model is instructed to assess how suitable the applicant is for that specific job and assign a score ($Y$) from 0 (least suitable) to 100 (most suitable).
We note that
our framework can also support alternative scoring rules,
such as the cosine similarity between embeddings of job descriptions and resumes, adopted by Resume Matcher\footnote{\url{https://github.com/srbhr/Resume-Matcher}}.

\paragraph{Evaluation Protocol.}
For LLM-based evaluation, we consider a practical setting in which the demographic information is not explicitly present in resumes.
For VLM-based evaluation, we consider two
scenarios: one using a text-only resume image as input, and another using a resume image with an attached profile photo.
Based on the model outputs, PSEs in \Cref{sec:pse_based_fair_eval} are estimated using DML-UCA with $L=5$.
We report 95\% confidence intervals obtained from 500 bootstrap samples.
PSE is considered \textit{significant} ($\neq 0$) if its confidence interval excludes zero and \textit{negligible} ($\simeq 0$) otherwise.
We consider
$5 \times 3 \times 2 = 30$ configurations spanning five occupations, three resume formats, two protected attributes and evaluate four models for each configuration, resulting in total 120 evaluation cases.
Full results are reported in \Cref{sec:full_additional_results}.

\subsection{When Do Different Pathways Matter?}
\label{sec:pse_eval_results}
To understand how demographic attributes influence model decisions through different causal pathways, we analyze the estimated causal effects across all experimental configurations.
Based on the PSEs, we identify several representative discrimination scenarios, presented in \Cref{fig:cases},
corresponding to a distinct auditing interpretation regarding the role of protected attributes in scoring.

\begin{figure}[t]
\centering

\begin{subfigure}{0.49\linewidth}
\centering
\includegraphics[width=\linewidth]{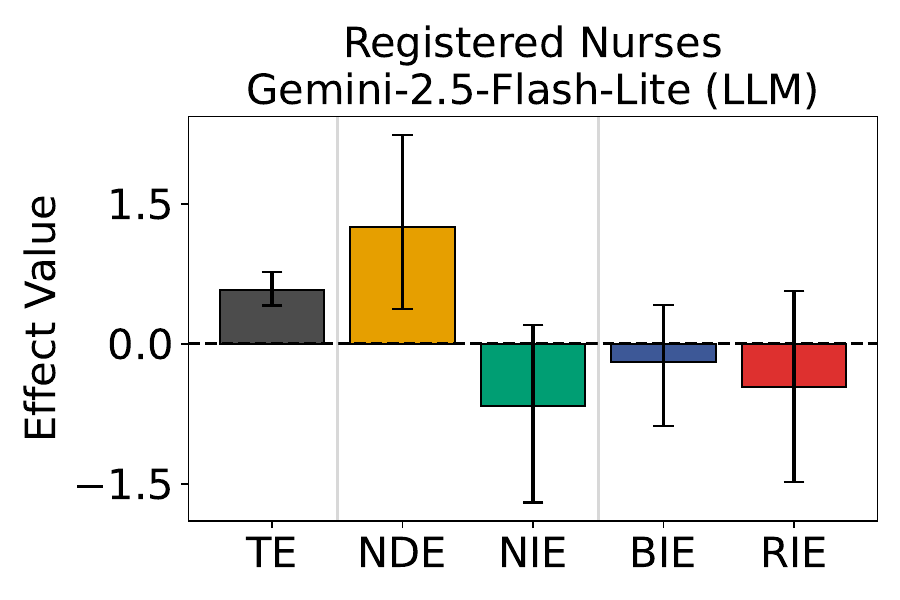}
\caption{Case 1}
\end{subfigure}
\hfill
\begin{subfigure}{0.49\linewidth}
\centering
\includegraphics[width=\linewidth]{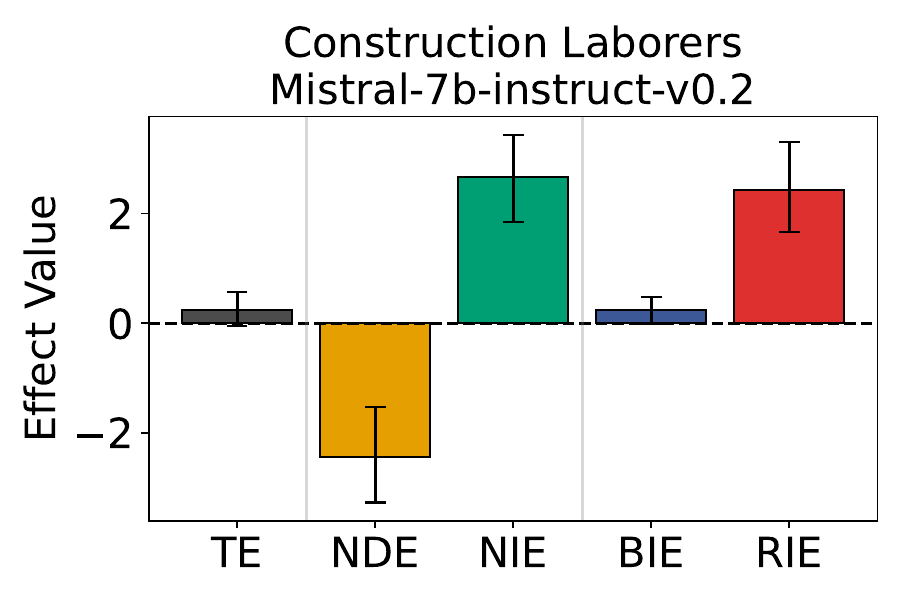}
\caption{Case 2}
\end{subfigure}

\begin{subfigure}{0.49\linewidth}
\centering
\includegraphics[width=\linewidth]{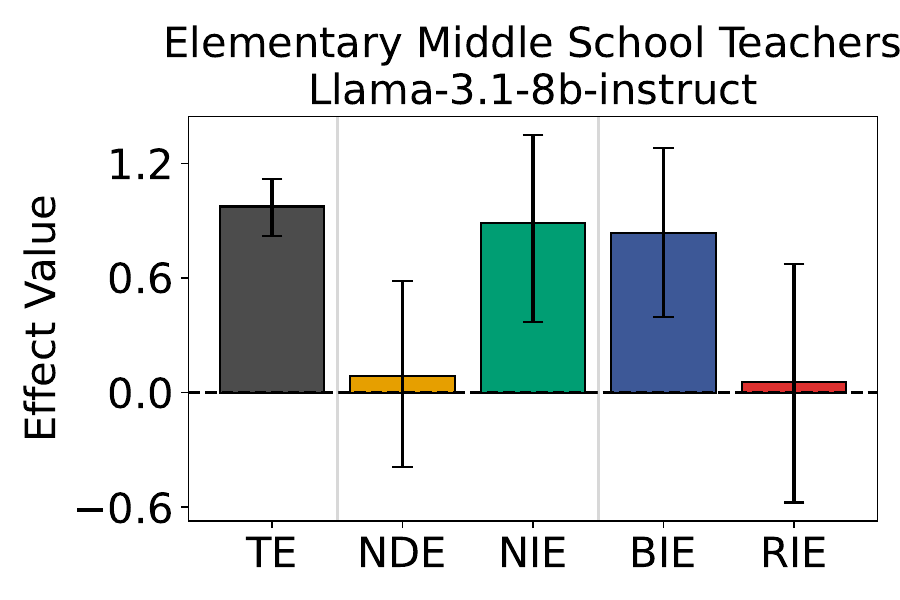}
\caption{Case 3}
\end{subfigure}
\hfill
\begin{subfigure}{0.49\linewidth}
\centering
\includegraphics[width=\linewidth]{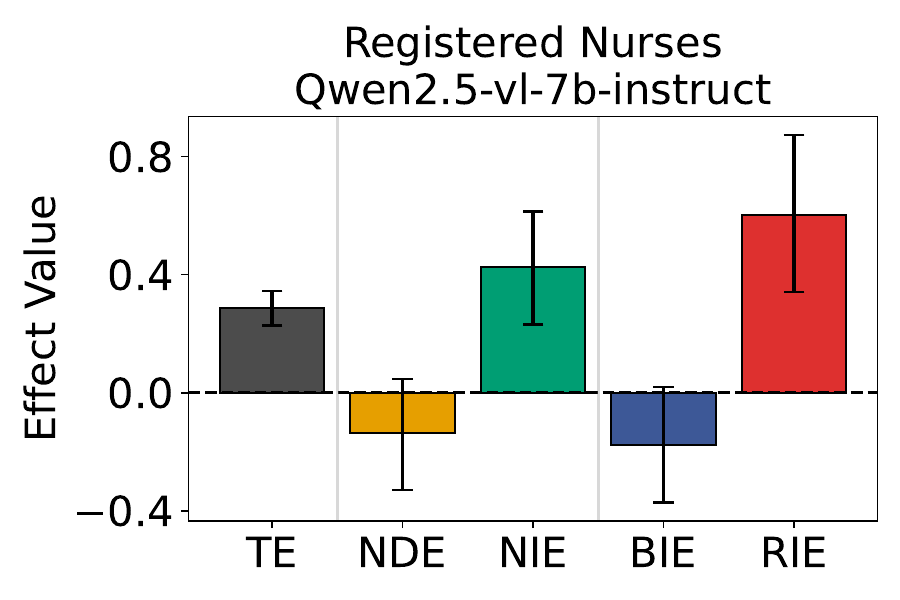}
\caption{Case 4}
\end{subfigure}

\begin{subfigure}{0.49\linewidth}
\centering
\includegraphics[width=\linewidth]{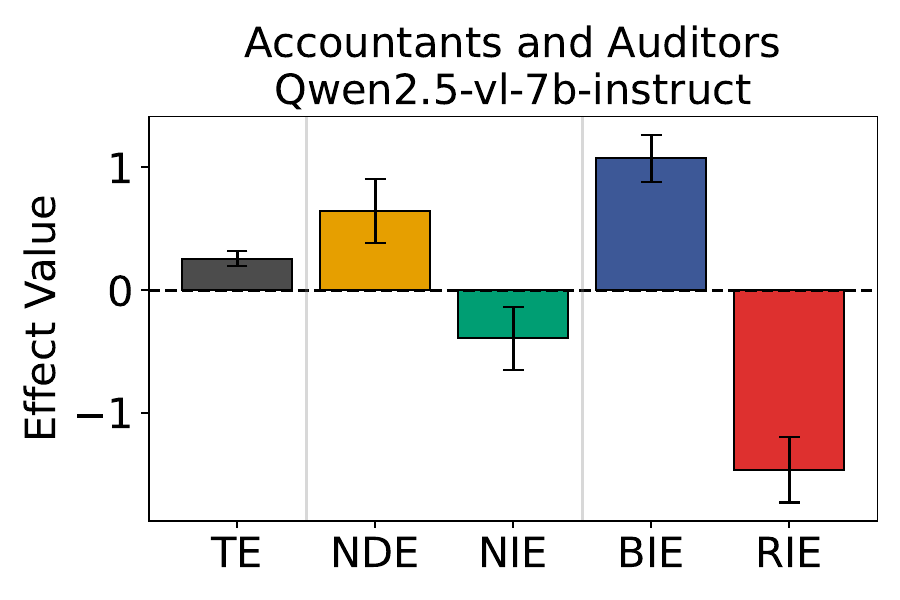}
\caption{Case 5 (cancellation)}
\end{subfigure}
\hfill
\begin{subfigure}{0.49\linewidth}
\centering
\includegraphics[width=\linewidth]{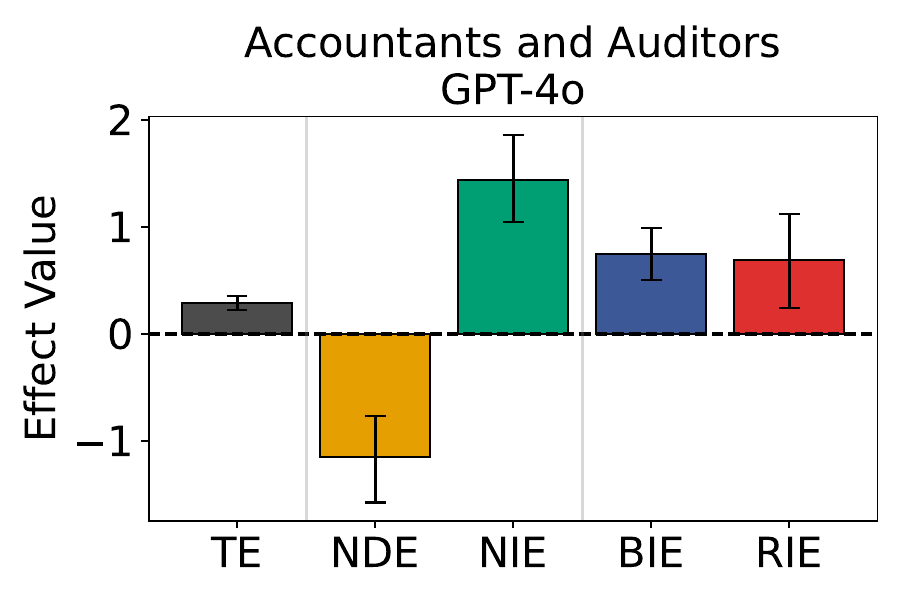}
\caption{Case 5 (amplification)}
\end{subfigure}
\caption{Representative examples of five cases based on our causal decomposition.
Error bars denote 95\% confidence intervals; an effect is significant when its interval excludes zero.
}
\label{fig:cases}
\end{figure}

\paragraph{Case 1: Direct Discrimination without Explicit Demographic Information.}
In this case, \textcolor{myorange}{NDE} $\neq 0$ even though demographic information is \textit{not} explicitly present in the resumes. This indicates that the model can infer protected attributes from other variables in the resume, such as names or education levels, and directly incorporate this inferred information into the scoring process.
Out of the 80 evaluation cases, in which protected attributes are absent from resumes, 49 fall into this category. This finding suggests that simply removing explicit demographic information from resumes is insufficient to prevent discrimination, as models can still infer such attributes and use them in their decisions.

\paragraph{Case 2: Discrimination Masked by Cancellation.}
In this case, TE $\simeq 0$ despite \textcolor{myorange}{NDE} $\neq 0$ and \textcolor{mygreen}{NIE} $\neq 0$, with \textcolor{myorange}{NDE} $\cdot$ \textcolor{mygreen}{NIE} $< 0$.
That is, the direct and indirect effects operate in opposite directions and largely cancel out, making TE $\simeq 0$.
Therefore, relying solely on TE, or equivalently, on outcome-based fairness metrics such as disparate impact, would lead an auditor to incorrectly conclude that no discrimination is present and hence the system is fair, while our framework reveals that substantial causal effects still exist.
We observe that 6 out of 120 evaluation cases fall into this case.

\paragraph{Case 3: Disparities Driven by Legitimate Qualifications.}
In this case, \textcolor{myorange}{NDE} $\simeq 0$ with \textcolor{myred}{RIE} $ \simeq 0$, yet \textcolor{myblue}{BIE} $\neq 0$, indicating that the observed disparity is transmitted primarily through job-relevant qualifications such as education levels and experience years. Under the disparate impact doctrine, such disparities may be considered legally defensible if the employer can demonstrate business necessity.
This case highlights the importance of distinguishing qualification-mediated disparities from proxy-based discrimination when auditing hiring systems. 
6 out of 120 evaluation cases fall into this case.

\paragraph{Case 4: Discrimination through Demographic Proxies.}
In this case, \textcolor{myorange}{NDE} $\simeq 0$ and \textcolor{myblue}{BIE} $\simeq 0$, yet \textcolor{myred}{RIE} $\neq 0$, meaning the disparity arises through demographic proxy variables, such as names and address.
These variables can encode demographic information of protected attributes, enabling the model to indirectly incorporate those signals into the scoring process.
This pathway is legally impermissible under Title VII, as it corresponds to the redlining doctrine: using geography- or name-based signals as proxies for protected attributes. Our decomposition allows auditors to isolate and identify this pathway regardless of whether the total effect is large or small.
We observe that 6 out of 120 evaluation cases fall into this case.

\paragraph{Case 5: Disparities Driven by Mixed Mediation Pathways.}
In this case, both \textcolor{myblue}{BIE} and \textcolor{myred}{RIE} are significant (\textit{i.e.,} \textcolor{myblue}{BIE} $\neq 0$ and \textcolor{myred}{RIE} $\neq$ 0), contributing to \textcolor{mygreen}{NIE}. This suggests that when the model assigns scores to resumes, it utilizes both job-related qualifications and demographic proxy signals simultaneously. Notably, these pathways can amplify or cancel out the magnitude of \textcolor{mygreen}{NIE}, thereby potentially obscuring any conflicting mechanisms.
This case demonstrates the significance of decomposing mediators into qualification-related ($\mathbf{B}$) and proxy-related components ($\mathbf{R}$), as proposed by our framework. If \textcolor{myblue}{BIE} and \textcolor{myred}{RIE} are not separated, various causal mechanisms may not be distinguishable at the level of \textcolor{mygreen}{NIE}, which could limit the auditor's ability to determine whether the observed gap arises from legitimate qualification factors or from discriminatory proxy variables.
We observe that 53 out of 120 evaluation cases fall into this case.

\subsection{Effect of Profile Photos on VLM-based Resume Scoring}
\label{sec:w_img_results}

\begin{figure}[t]
\centering

\begin{subfigure}{0.75\linewidth}
\includegraphics[width=\linewidth]{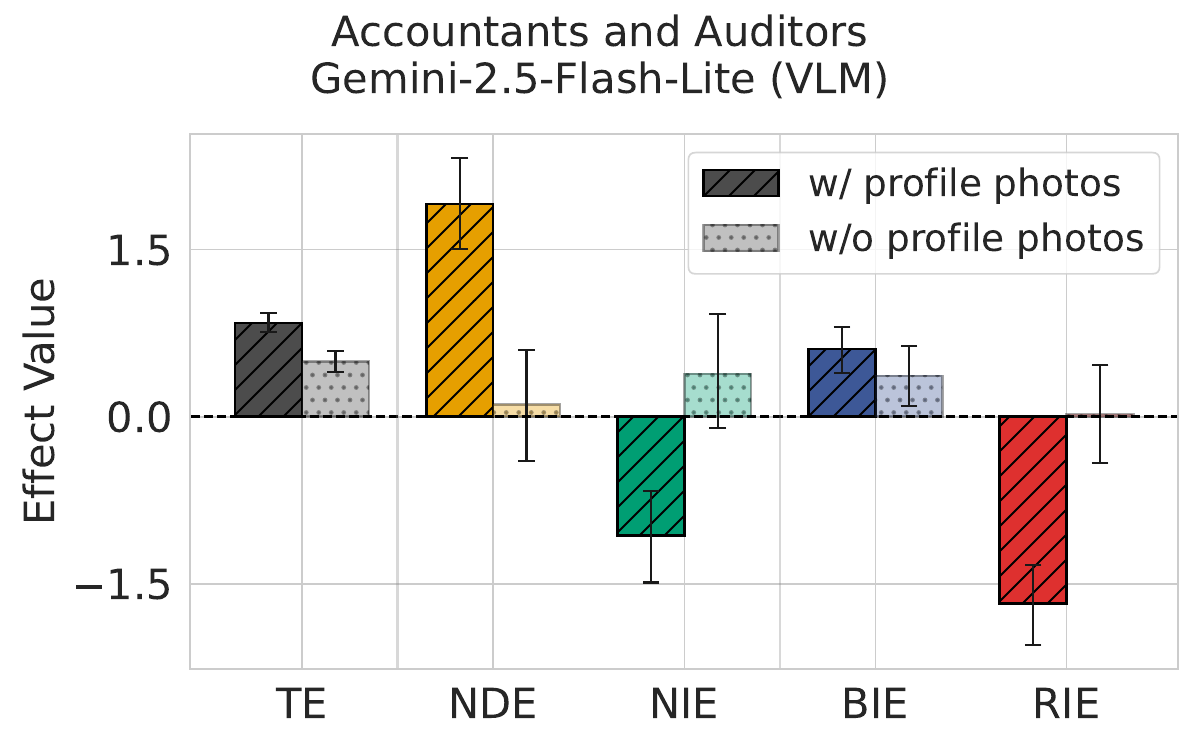}
\end{subfigure}

\begin{subfigure}{0.75\linewidth}
\includegraphics[width=\linewidth]{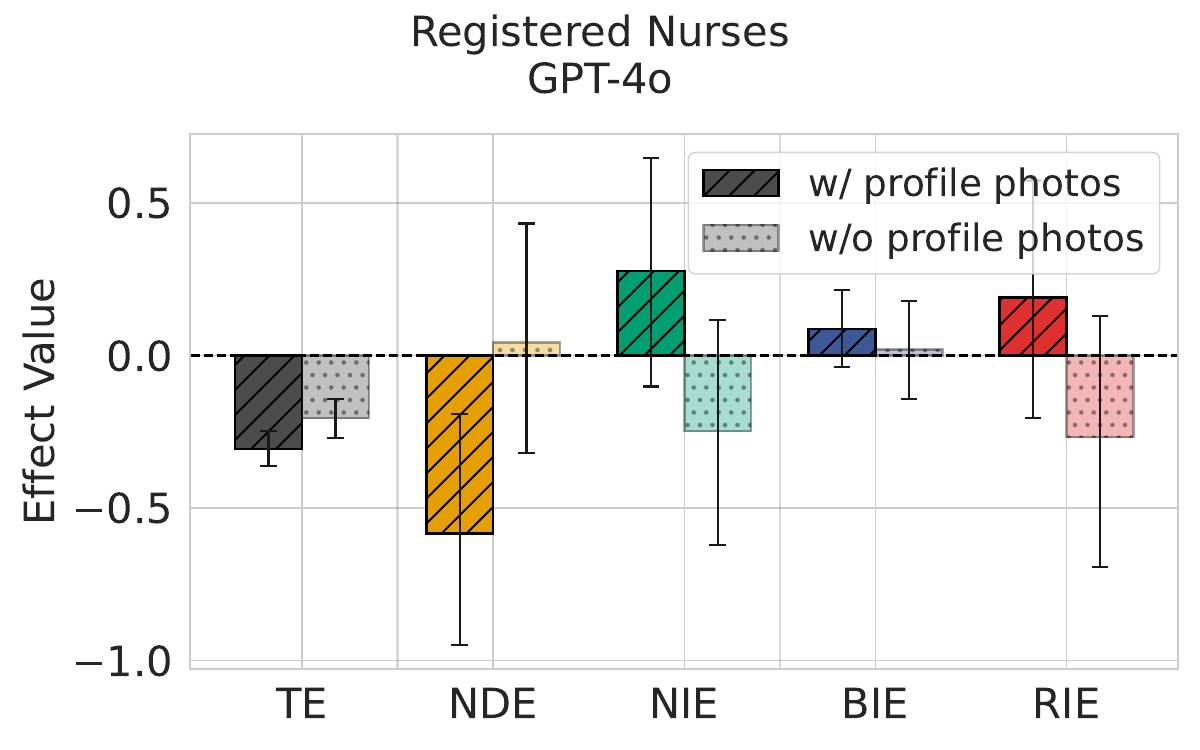}
\end{subfigure}

\caption{PSEs in VLM-based resume scoring.
Error bars denote 95\% confidence intervals; an effect is significant when its interval excludes zero.
}

\label{fig:profile_image_effect}
\end{figure}

Since image-based resumes may include profile photos in realistic scenarios, we generate synthetic profile photos conditioned on each applicant's demographic attributes and attach them to the resumes. Based on two variants of resumes, image resumes with and without profile photos, we evaluate resume scoring VLM models and estimate the corresponding PSEs. This allows us to examine how the presence of profile photos influences the causal pathways through which demographic attributes affect model scores.
Two representative examples are shown in \Cref{fig:profile_image_effect}.
When profile photos are included, we observe that the magnitude of \textcolor{myorange}{NDE} increases noticeably.
Across 20 pairs of evaluation cases, 8 exhibit this pattern.
This suggests that profile photos make demographic attributes more directly accessible to the model, enabling it to infer protected attributes more easily and reflect them in the scoring process, highlighting an additional risk in image-based resume screening.
Organizations deploying VLM-based resume screening systems should therefore carefully assess whether profile photos are necessary for the hiring task.
When visual information is incorporated, additional auditing or monitoring procedures may be required to ensure that demographic signals do not introduce unintended bias into automated scoring pipelines.

\section{Conclusion}
In this work, we presented \textbf{PopResume}, a population-representative
resume dataset and a PSE-based causal fairness auditing framework
for LLM/VLM-based resume screening systems. Using PopResume and the newly introduced PSEs, BIE and RIE, our framework enables auditors to
distinguish legally permissible from impermissible sources of
disparity, a distinction that outcome-level metrics fundamentally
cannot make. Through extensive evaluations across diverse configurations,
we identified five representative discrimination patterns that
aggregate fairness metrics fail to capture.
Our results further demonstrate that profile photos can introduce an additional
direct discrimination risk in VLM-based evaluation. We hope
PopResume serves as a foundation for causally-grounded, legally-informed auditing of AI-assisted resume screening systems.
\section*{Limitations}
{
Our dataset is constructed primarily from U.S. population statistics and therefore may not fully reflect demographic distributions, labor-market structures, or legal and societal norms in other countries. In particular, the distinction between variables considered legitimate job-related factors and those regarded as impermissible proxies may vary across jurisdictions and application contexts. Accordingly, the specific dataset and variable categorization used in this work should not be directly assumed to generalize beyond the U.S. setting.
Our construction pipeline itself, however, is not tied to U.S. sources: comparable population surveys, such as the EU Labour Force Survey (EU-LFS), could support the construction of analogous datasets for other regions.
}

In addition, although our dataset is designed to preserve realistic relationships between attributes based on population statistics, the resumes themselves are synthetically generated using rule-based procedures. As a result, certain aspects of real-world resumes, such as nuanced career trajectories, writing styles, or unstructured information, may not be fully captured.

{
Another limitation concerns the validation of the evaluation framework itself. Establishing whether our PSE-based framework accurately recovers the true degree of bias in an LLM/VLM is inherently challenging, as the ground-truth (un)fairness of a model in realistic hiring settings is generally not directly observable or well defined as a single measurable quantity. One possible direction for future work is to construct controlled models with known bias injected through specific causal pathways, \textit{i.e.}, through $\mathbf{R}$ (\textit{e.g.}, name or location) or $\mathbf{B}$ (\textit{e.g.}, education or experience), and then evaluate whether our estimated PSEs recover the injected effects, which we leave for future work.
}

Despite these limitations, our approach provides a controlled and transparent environment that enables causal pathway–based fairness evaluation, which is difficult to achieve using existing real-world or manually perturbed benchmarks. We therefore view our dataset as a complementary tool that facilitates systematic analysis of fairness in LLM/VLM-based resume screening systems.

\section*{Ethical Considerations}
\label{ethical_consideration}

In this work, we analyze fairness in resume screening systems through the lens of causal pathways, distinguishing between redlining-related and business-necessity-related pathways following legal and policy discussions in employment discrimination.
However, determining whether a particular pathway should be considered fair or unfair is ultimately a normative question that depends on legal, societal, and contextual considerations.
The categorization adopted in this work should therefore not be interpreted as a definitive or universally applicable judgment about fairness.

Instead, our goal is to provide an analytical framework that enables auditors to examine how different causal pathways contribute to observed disparities in automated hiring systems.
By decomposing the total effect into interpretable components, our approach aims to support more transparent and structured discussions about potential sources of discrimination.
We believe that such analyses can serve as a useful tool for researchers, practitioners, and policymakers when assessing the fairness implications of AI-assisted hiring systems, even though the final determination of fairness remains context-dependent.

We also note that the specific categorization of resume attributes into $\mathbf{B}$ and $\mathbf{R}$ reflects one operationalization motivated by legal and policy discussions, and may vary across jurisdictions, industries, or hiring contexts. Our framework does not require a fixed categorization. Rather, alternative attribute groupings can be defined and their causal effects assessed accordingly.
\section*{Acknowledgments}
This work was supported in part by the National Research Foundation of Korea (NRF) grant
[No. RS-2025-02263628], the Institute of Information \& communications Technology Planning
\& Evaluation (IITP) grants [RS-2021-II212068, RS-2022-II220113, RS-2022-II220959, RS-2021-II211343], and the BK21 FOUR Education and
Research Program for Future ICT Pioneers (Seoul
National University), funded by the Korean government. It was also supported by AOARD Grant No. FA2386-25-1-4015 and Hyundai Motor Chung Mong-Koo Foundation.

\bibliography{custom}

\appendix
\section*{The Use of Large Language Models (LLMs)}
We utilized LLMs for the purpose of polishing our manuscript only.

\section{Related Works}

\subsection{Fairness in LLM-based Hiring Systems}
The use of LLMs for resume screening and candidate evaluation has rapidly increased, enabling automated assessment of applicant suitability in hiring pipelines. However, recent studies have raised concerns about fairness and demographic bias in such systems, showing that LLM outputs may vary depending on protected attributes such as gender and race.

\citet{hu2025fairwork, wang2024jobfair} investigate demographic bias in LLM-based hiring by explicitly representing protected attributes and then perturbing them to observe the resulting changes in scores or rankings. To better reflect practical settings where demographic attributes are not explicitly stated, \citet{nghiem2024you, iso2025evaluating, armstrong2024silicon, wilson2024gender} inject and perturb protected attribute information through names that are typically associated with specific genders or races.

While these approaches can reveal the presence of demographic disparities, they do not distinguish the causal pathways through which such disparities arise. Moreover, because demographic information is manually injected into each resume, the original relationships between protected attributes and other variables are disrupted, making causal-framework-based evaluation infeasible.

Our work addresses these limitations by introducing a causal framework that enables path-specific analysis of demographic effects in LLM- and VLM-based resume scoring systems, while generating resume datasets that preserve the natural relationships between protected attributes and other variables.

\subsection{PSE-Based Algorithmic Fairness}
Path-specific effect (PSE)–based fairness has emerged as a flexible paradigm for auditing and mitigating discrimination in algorithmic decision-making \cite{zhang2018fairness, chiappa2019path, plecko2024causal}. The central idea is to distinguish fair from unfair causal pathways linking a sensitive attribute to a decision: effects transmitted through legitimate mediators are considered permissible, whereas those propagating through designated unfair pathways are restricted or removed.

Within this framework, \citet{zhang2025path} recently examined how race-mediated discrepancies in pulse oximetry measurements propagate through clinical decision pathways in intensive care units (ICUs). However, to the best of our knowledge, this framework has not yet been applied to auditing the fairness of LLM/VLM-based resume screening systems.

\section{Additional Details on the Resume Dataset Construction}
\label{sec:resume_data_detail}
\subsection{Instantiation Rules}
\label{sec:app_instantiation}
Each structured applicant profile is deterministically converted into a natural-language resume using rule-based instantiation procedures. The generation process constructs the resume sections--education, work history, and contact information--directly from the structured variables $(X,\mathbf{Z},\mathbf{B},\mathbf{R})$.

\paragraph{Education.}
Education entries are generated from the variable \textit{edu\_level}. 
The number of education entries and bullet descriptions depend on the education level. 
For example, Bachelor's, Associate's, and high school levels produce a single entry, while advanced degrees such as Master's or Doctorate produce two entries, consisting of a Bachelor's degree followed by the higher degree. 
Institution names are sampled from a pool of fictional schools corresponding to each degree level, and majors are sampled from occupation-specific major pools. 
Graduation years are determined based on age and degree type to ensure consistency with realistic educational timelines.

\paragraph{Work History.}
Work history entries are generated from the predicted years of experience (\textit{pred\_exp}). 
The number of work roles is determined by a piecewise rule based on experience years. For example, individuals with less than three years of experience are assigned a single role, while individuals with more than ten years of experience may have up to four roles. 
Employment periods are constructed sequentially starting from the end of education, and role durations are distributed across the total experience length. 
Each role includes a company name sampled from occupation-specific company pools and responsibility descriptions sampled from job- and seniority-specific bullet pools. 
Role seniority (junior, mid, senior) is determined based on the temporal order of positions.

\paragraph{Name, Address, and Contact Information.}
Names and addresses originate from the structured variables sampled during dataset construction. 
State names are converted into their corresponding two-letter USPS abbreviations when rendered in resumes. 
Phone numbers and email addresses are synthetically generated and do not correspond to real individuals. 
Email addresses are constructed from the individual's name, and phone numbers follow standard U.S. phone number formatting.

\paragraph{Resume Formatting.}
The final resume text is generated by assembling the sections into a standardized format consisting of a header (name, location, and contact information), followed by skills, work history, and education sections. 
This deterministic template ensures that variation in the generated resumes arises solely from the underlying structured attributes rather than uncontrolled textual randomness.

\subsection{Q-Q plot for Assumption 2}
We visualize the Q-Q plot to validate the normality assumption in Assumption 2. As shown in \Cref{fig:qq_plot}, the empirical quantiles closely follow the reference line, indicating that the normality assumption is reasonably satisfied.

\label{sec:qq_plot}
\begin{figure}[ht]
\centering
\includegraphics[width=\linewidth]{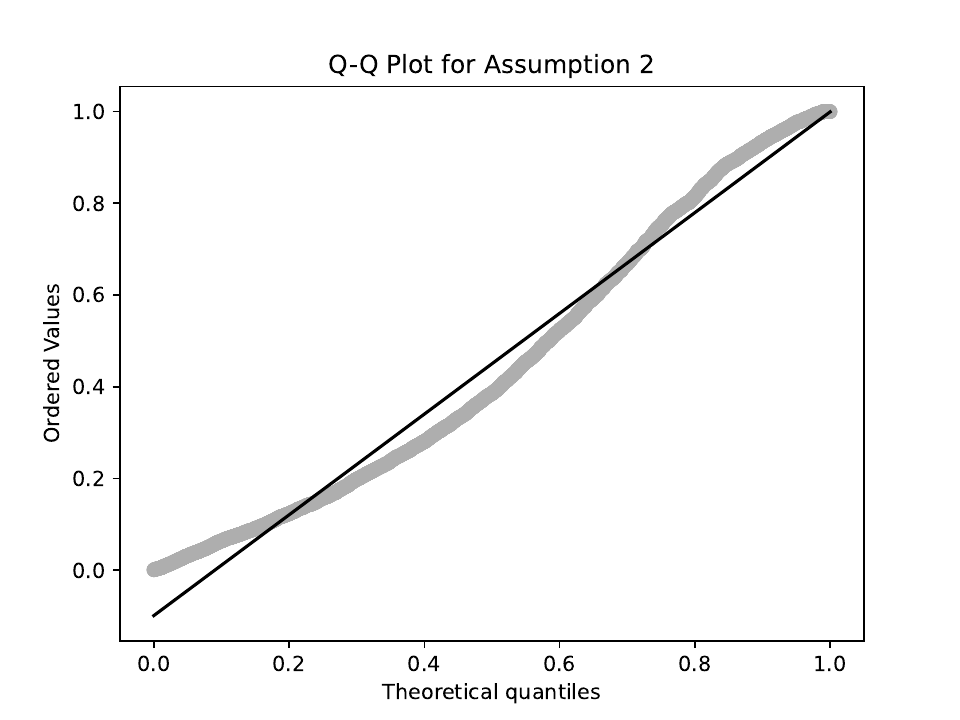}
\caption{Q-Q Plot for Assumption 2.}
\label{fig:qq_plot}
\end{figure}

\subsection{Detailed Dataset Statistics}
\label{sec:data_stat}

\begin{table}[t]
\centering
\caption{Number of structured applicant profiles generated for each occupation in PopResume.}
\label{tab:dataset_size}
\resizebox{\linewidth}{!}{%
\begin{tabular}{l r}
\toprule
Occupation & Num \\
\midrule
Registered Nurses & 17,632 \\
Elementary and Middle School Teachers & 16,188 \\
Software Developers & 13,370 \\
Accountants and Auditors & 7,784 \\
Construction Laborers & 5,910 \\
\midrule
\textbf{Total} & \textbf{60,884} \\
\bottomrule
\end{tabular}%
}
\end{table}
PopResume contains 60,884 structured applicant profiles spanning five occupations:
registered nurses, elementary and middle school teachers, software developers, accountants and auditors, and construction laborers.
\Cref{tab:dataset_size} summarizes the number of profiles generated for each occupation. Each profile is subsequently rendered into multiple resume formats used in the evaluation.

\Cref{fig:dataset_stats} summarizes key statistical properties of the constructed \textbf{PopResume} dataset. 
These statistics illustrate the demographic composition of the dataset.
(a) and (b) show the gender and race compositions across occupations. Both reflect the labor market patterns present in the underlying ACS PUMS microdata—for example, construction laborers are predominantly male, while nursing and teaching occupations show higher proportions of female.
(c) shows the distribution of education levels across occupations.
 The distributions reflect the patterns observed in the underlying ACS PUMS microdata. For example, software developers and accountants tend to have higher education levels, while occupations such as construction laborers include more individuals with lower formal education.
(d) shows the relationship between age and predicted years of work experience.
As expected, work experience generally increases with age. This pattern indicates that the PSID-based model generates realistic career trajectories that align with typical labor market patterns.
Taken together, these statistics confirm that \textbf{PopResume} preserves realistic demographic and qualification structures while enabling controlled causal analysis.

\subsection{Example Resume}
Example resumes from the constructed dataset are shown in \Cref{fig:resume_examples}. These examples illustrate the input formats used in our evaluation, including text resumes for LLM-based scoring and image resumes for VLM-based scoring with and without profile photos.

\subsection{Validation of Demographic Information in Profile Images}
\label{sec:app_image_check}
To check whether gender and race information were correctly reflected in the generated profile images, we employed CLIP for zero-shot classification. Specifically, gender was evaluated using two-class zero-shot classification with the prompts “a photo of a man” and “a photo of a woman.” Race was evaluated using the prompts “a photo of a white person,” “a photo of an Asian/Pacific Islander person,” and “a photo of a Black person.” Based on this procedure, we confirmed that all generated images are consistent with the information specified in the prompt.

For the evaluation of age, we used DeepFace\footnote{\url{https://github.com/serengil/deepface}}
 and measured the deviation between the age specified in the prompt and the age predicted by the model. The average deviation was 6.56 years. However, we note that age prediction from facial images is inherently nuanced and subject to considerable uncertainty; therefore, this value should be interpreted as a rough indicator rather than an exact measure of age consistency.
 
\begin{figure}[ht]
\centering

\begin{subfigure}[t]{0.98\linewidth}
\centering
\includegraphics[width=\linewidth]{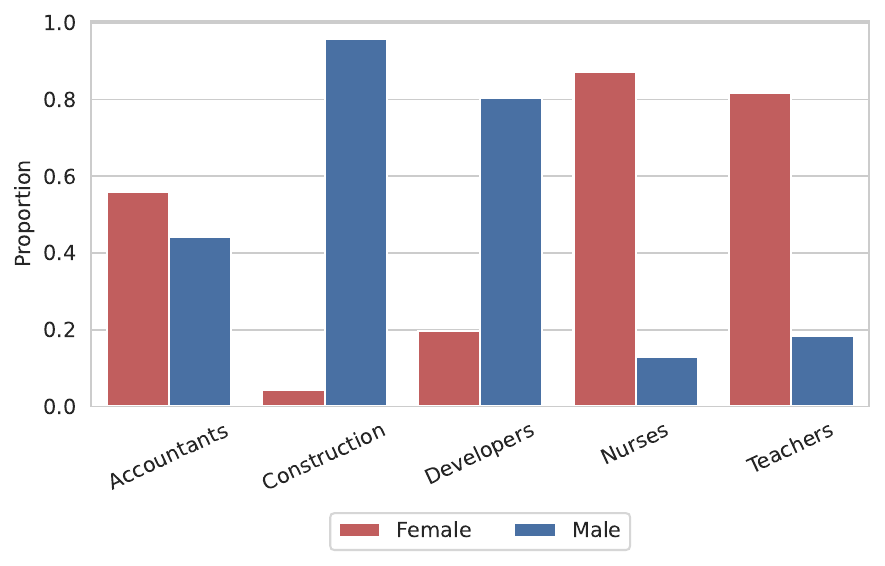}
\caption{Gender distribution by occupation}
\end{subfigure}
\hfill
\begin{subfigure}[t]{0.98\linewidth}
\centering
\includegraphics[width=\linewidth]{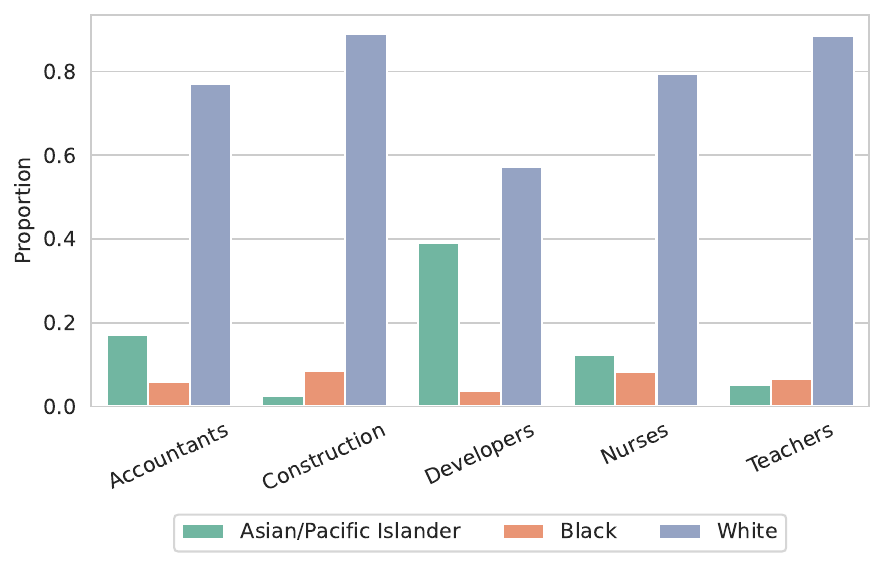}
\caption{Race distribution by occupation}
\end{subfigure}

\begin{subfigure}[t]{0.98\linewidth}
\centering
\includegraphics[width=\linewidth]{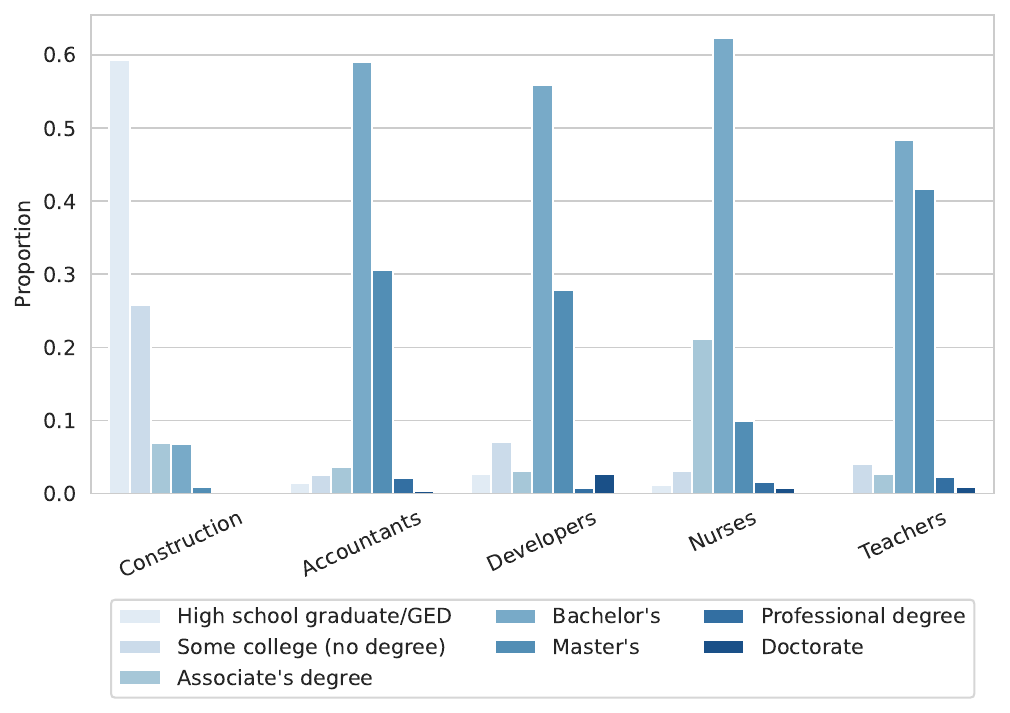}
\caption{Education level distribution by occupation}
\end{subfigure}
\hfill
\begin{subfigure}[t]{0.98\linewidth}
\centering
\includegraphics[width=\linewidth]{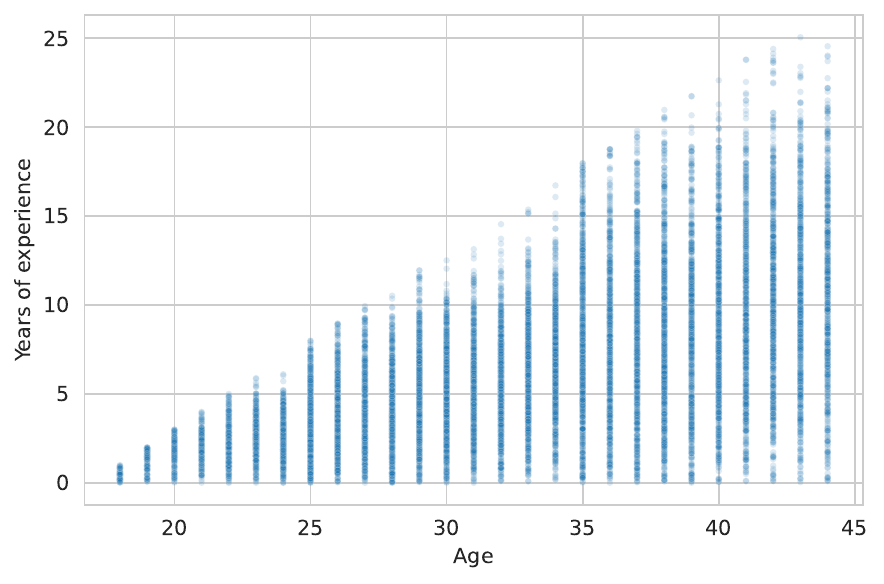}
\caption{Age vs experience. (predicted years of experience derived from PSID-based modeling.)}
\end{subfigure}

\caption{
Statistical characteristics of the PopResume dataset.
}
\label{fig:dataset_stats}

\end{figure}

\begin{figure*}[t]
\centering

\includegraphics[width=0.55\linewidth]{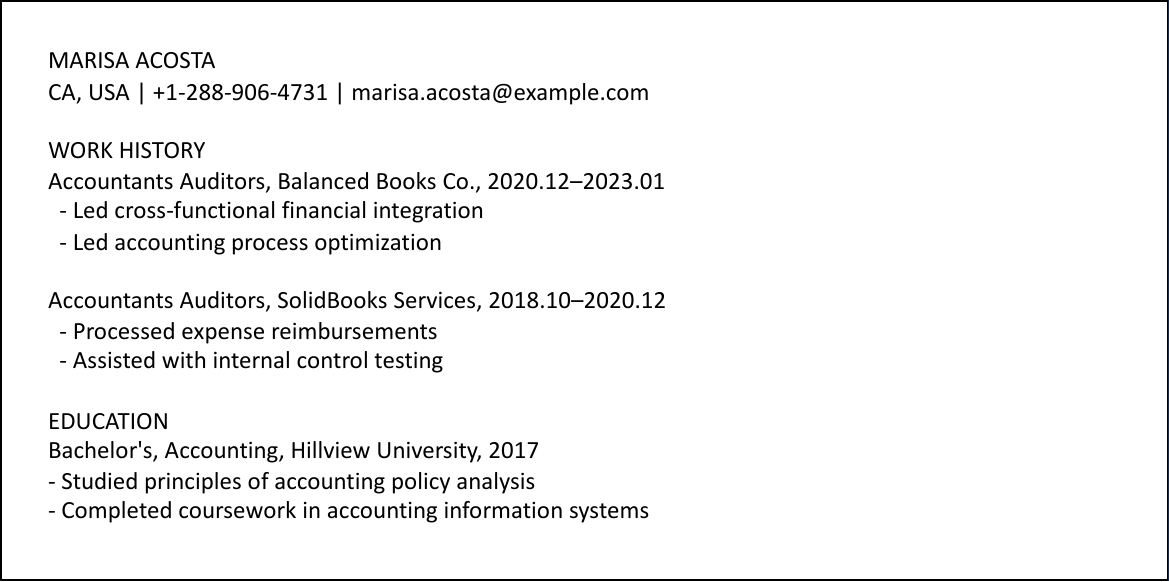}

(a) Text resume used for LLM-based scoring.

\includegraphics[width=0.55\linewidth]{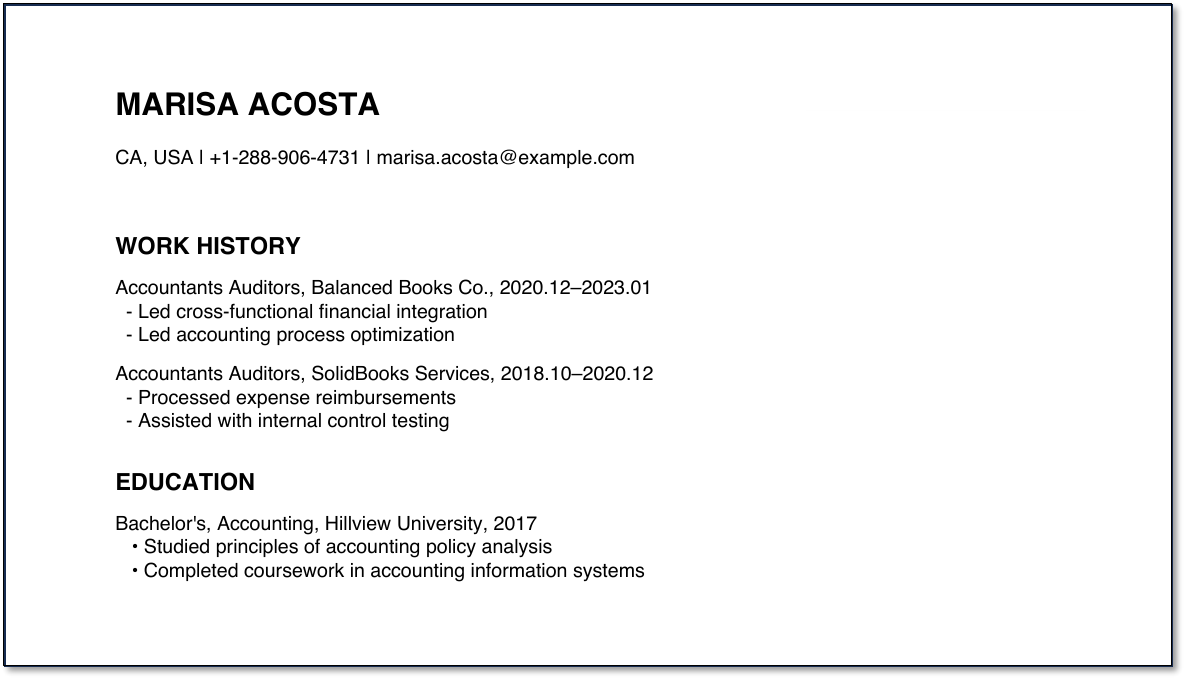}

(b) Image resume without a profile photo used for VLM-based scoring.

\includegraphics[width=0.55\linewidth]{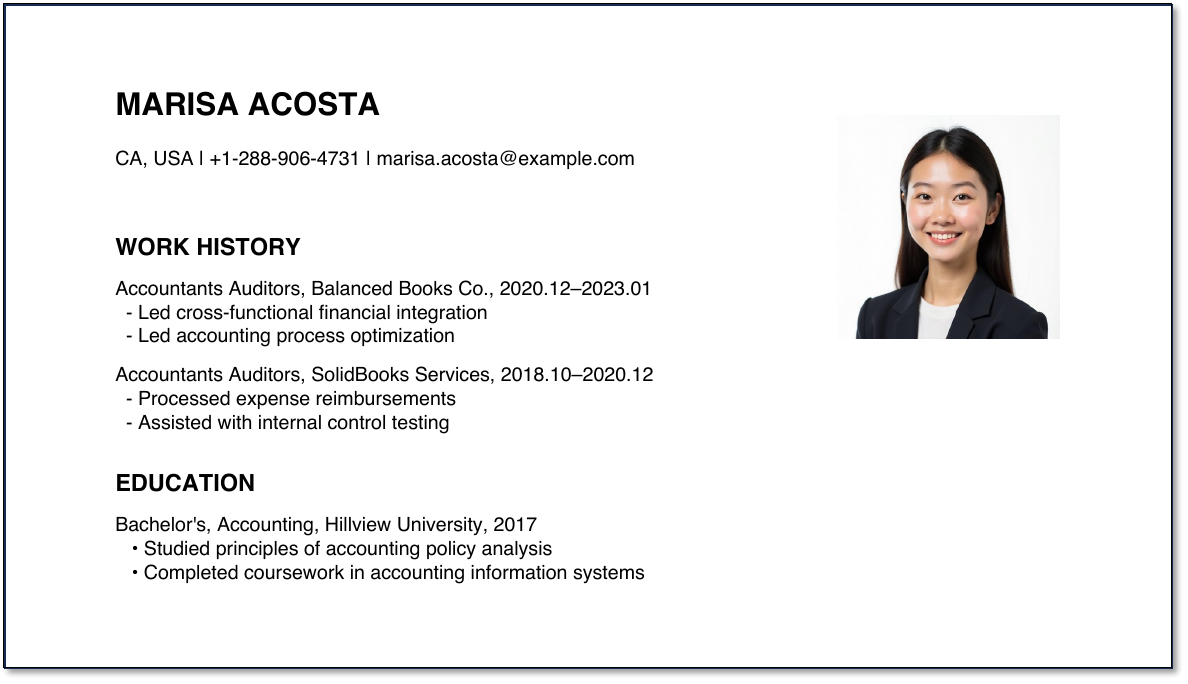}

(c) Image resume with a profile photo used for VLM-based scoring.

\caption{Example resumes from the constructed dataset.}
\label{fig:resume_examples}

\end{figure*}

\onecolumn
\clearpage
\section{Details on Estimation of Path-Specific Effects}
\label{sec:pse_estimation}
\subsection{Identifiability}
\label{sec:app_identi}
All terms used for the estimation are identifiable as follows, as proved by \citet{zhang2025path} and \citet{plecko2024causal}. 
    
\begin{align}
\mathbb{E}[Y_{x_0}]
&= \sum_{\mathbf{z}} \mathbb{E}[Y \mid x_0, \mathbf{z}] P(\mathbf{z}), \\
\mathbb{E}[Y_{x_1}]
&= \sum_{\mathbf{z}} \mathbb{E}[Y \mid x_1, \mathbf{z}] P(\mathbf{z}), \\
\mathbb{E}[Y_{x_1, \mathbf{W}_{x_0}}]
&= \sum_{\mathbf{w}, \mathbf{z}}
\begin{aligned}[t]
\mathbb{E}[Y \mid x_1, \mathbf{w}, \mathbf{z}]
\cdot P(\mathbf{w} \mid x_0, \mathbf{z}) P(\mathbf{z}),
\end{aligned} \\
\mathbb{E}[Y_{x_1,\mathbf{B}_{x_0,\mathbf{R}_{x_1}},\mathbf{R}_{x_1}}]
&= \sum_{\mathbf{b},\mathbf{r},\mathbf{z}}
\begin{aligned}[t]
\mathbb{E}[Y \mid x_1, \mathbf{b}, \mathbf{r}, \mathbf{z}]
\cdot P(\mathbf{b} \mid x_0, \mathbf{r}, \mathbf{z})
P(\mathbf{r} \mid x_1, \mathbf{z}) P(\mathbf{z}).
\end{aligned} \\
\mathbb{E}[Y_{x_1,\mathbf{B}_{x_1,\mathbf{R}_{x_0}},\mathbf{R}_{x_0}}]
&= \sum_{\mathbf{b},\mathbf{r},\mathbf{z}}
\begin{aligned}[t]
\mathbb{E}[Y \mid x_1, \mathbf{b}, \mathbf{r}, \mathbf{z}] 
\cdot P(\mathbf{b} \mid x_1, \mathbf{r}, \mathbf{z})
P(\mathbf{r} \mid x_0, \mathbf{z}) P(\mathbf{z}).
\label{eq:ours_identi}
\end{aligned}
\end{align}

\subsection{Parametrization}
\label{sec:app_paramet}
\begin{lemma}[Parametrization]
\begin{equation}
\label{eq:def_varphi}
\begin{aligned}
\varphi(\mathbf{V};\bm{\mu},\bm{\pi})
\triangleq\;&
\pi^3\{Y-\mu^3(\mathbf{B,R}, X, \mathbf{Z})\} \\
&+ \pi^2\{\mu^3(\mathbf{B,R},x_1,\mathbf{Z})-\mu^2(\mathbf{R},X,\mathbf{Z})\} \\
&+ \pi^1\{\mu^2(\mathbf{R},x_1,\mathbf{Z})-\mu^1(X,\mathbf{Z})\} \\
&+ \mu^1(x_0,\mathbf{Z}).
\end{aligned}
\end{equation}
is a valid representation in the sense that $\mathbb{E}[\varphi(\mathbf{V};\bm{\mu},\bm{\pi})]=$\Cref{eq:ours_identi}.
\end{lemma}
\begin{proof}
We note that the proof is largely adopted from Section A.3 in \citet{zhang2025path}.
Further, we note that the $\pi^3(\mathbf{B,R}, X,\mathbf{Z})=\frac{P(X=x_0|\mathbf{R,Z})}{P(X|\mathbf{R,Z})} \frac{\mathds{1}[X={x_1}]}{P(X=x_0|\mathbf{Z})}$ can be represented as $\frac{P(\mathbf{R}\mid X=x_0,\mathbf{Z})\mathds{1}[X=x_1]}{P(\mathbf{R}\mid X,\mathbf{Z})P(X\mid \mathbf{Z})}$,
which follows directly from Bayes' rule.

\begin{equation}
\begin{aligned}
\mu^2(\mathbf{R},X,\mathbf{Z})
&\triangleq \mathbb{E}[\mu^3(\mathbf{B,R},x_1,\mathbf{Z})\mid \mathbf{R},X,\mathbf{Z}] \\
&= \sum_{\mathbf{b}} \mu^3(\mathbf{b},\mathbf{R},x_1,\mathbf{Z}) P(\mathbf{b}\mid \mathbf{R},X,\mathbf{Z}) \\
&= \sum_\mathbf{b} \mathbb{E}[Y\mid \mathbf{b},\mathbf{R},x_1,\mathbf{Z}] P(\mathbf{b}\mid \mathbf{R},X,\mathbf{Z}).
\end{aligned}
\end{equation}

and
\begin{equation}
\begin{aligned}
\mu^1(X,\mathbf{Z})
&\triangleq \mathbb{E}[\mu^2(\mathbf{R},x_1,\mathbf{Z})\mid X,\mathbf{Z}] \\
&= \sum_{\mathbf{r}} \mu^2(\mathbf{r},x_1,\mathbf{Z}) P(\mathbf{r}\mid X,\mathbf{Z}) \\
&= \sum_{\mathbf{b,r}} \mathbb{E}[Y\mid \mathbf{b,r},x_1,\mathbf{Z}]
   P(\mathbf{b}\mid \mathbf{r},x_1,\mathbf{Z})P(\mathbf{r}\mid X,\mathbf{Z}).
\end{aligned}
\end{equation}

Therefore,
\begin{equation}
\mathbb{E}[\mu^1(x_0,Z)] = \text{\Cref{eq:ours_identi}}.
\end{equation}

Furthermore,
\begin{equation}
\begin{aligned}
\mathbb{E}[\pi^3(\mathbf{B,R},X,\mathbf{Z})Y]
&= \mathbb{E}[\pi^3(\mathbf{B,R},X,\mathbf{Z})\mu^3(\mathbf{B,R},X,\mathbf{Z})] \\
&= \sum_{\mathbf{b,r},x,\mathbf{z}} \mu^3(\mathbf{b,r},x,\mathbf{z}) \frac{P(\mathbf{r}\mid x_0,\mathbf{z})\mathds{1}[x=x_1]}
        {P(\mathbf{r}\mid x,\mathbf{z})P(x\mid \mathbf{z})}P(\mathbf{b,r},x,\mathbf{z})\\
&= \sum_{\mathbf{b,r,z}} \mu^3(\mathbf{b,r},x_1,\mathbf{z})
   P(\mathbf{b}\mid x_1,\mathbf{r,z})P(\mathbf{r}\mid x_0,\mathbf{z})P(\mathbf{z}) \\
&= \text{\Cref{eq:ours_identi}}.
\end{aligned}
\end{equation}

Also,
\begin{equation}
\begin{aligned}
\mathbb{E}[\pi^2(\mathbf{R},X,\mathbf{Z})\mu^2(\mathbf{R},X,\mathbf{Z})]
&= \sum_{\mathbf{r},x,\mathbf{z}} \mu^2(\mathbf{r},x,\mathbf{z})
   \frac{P(\mathbf{r}\mid x_0,\mathbf{z})\mathds{1}[x=x_1]}
        {P(\mathbf{r}\mid x,\mathbf{z})P(x\mid \mathbf{z})} P(\mathbf{r},x,\mathbf{z}) \\
&= \sum_{\mathbf{r,z}} \mu^2(\mathbf{r},x_1,\mathbf{z})P(\mathbf{r}\mid x_0,\mathbf{z})P(\mathbf{z}) \\
&= \sum_{\mathbf{r,z}}\sum_\mathbf{b} \mathbb{E}[Y\mid \mathbf{b,r},x_1,\mathbf{z}]
   P(\mathbf{b}\mid \mathbf{r},x_1,\mathbf{z})P(\mathbf{r}\mid x_0,\mathbf{z})P(\mathbf{z}) \\
&= \text{\Cref{eq:ours_identi}}.
\end{aligned}
\end{equation}

Finally,
\begin{equation}
\mathbb{E}[\pi^1(X,\mathbf{Z})\mu^1(X,\mathbf{Z})]
= \mathbb{E}[\mu^1(x_0,\mathbf{Z})]
= \text{\Cref{eq:ours_identi}}.
\end{equation}

\end{proof}
\subsection{Doubly Robustness Property}
\label{sec:app_doubly_robust}
\begin{lemma}[Doubly Robustness]
For any arbitrary 
$\hat{\bm{\mu}}, \hat{\bm{\pi}}$, the $\varphi$ in \Cref{eq:def_varphi} satisfies
\begin{equation}
    \mathbb{E}[\varphi(\mathbf{V};\bm{\mu},\bm{\pi})] - \mathbb{E}[\varphi(\mathbf{V};\hat{\bm{\mu}},\hat{\bm{\pi}})] = \Sigma_{i=1}^3\mathcal{O}_P(||\mu^i-\hat{\mu^i}||_P||\pi^i-\hat{\pi^i}||_P)
\end{equation}
\end{lemma}
\begin{proof}
Following the Section A.4 in \citet{zhang2025path}, it is enough to show that two equalities. One is 
\begin{equation}
    \mathbb{E}[\pi^3(\mathbf{B,R},X,\mathbf{Z})\hat{\mu^3}(\mathbf{B,R},X,\mathbf{Z})]=\mathbb{E}[\pi^2(\mathbf{R},X,\mathbf{Z})\hat{\mu^2_*}(\mathbf{R},X,\mathbf{Z})]
\end{equation}
where $\hat{\mu_*^2}\triangleq\mathbb{E}[\hat{\mu^3}(\mathbf{B,R},x_1,\mathbf{Z})|\mathbf{R},X,\mathbf{Z}]$ for any fixed $\hat{\mu^3}$.
The equality holds since
\begin{equation}
\begin{aligned}
    &\mathbb{E}[\pi^2(\mathbf{R},X,\mathbf{Z})\hat{\mu^2_*}(\mathbf{R},X,\mathbf{Z})] \\
    &= \mathbb{E}[\pi^2(\mathbf{R},X,\mathbf{Z})\mathbb{E}[\hat{\mu^3}(\mathbf{B,R},x_1,\mathbf{Z})|\mathbf{R},X,\mathbf{Z}]] \\
    &=\Sigma_{\mathbf{b,r},x,\mathbf{z}}\hat{\mu^3}(\mathbf{b,r},x_1,\mathbf{z})P(\mathbf{b}|\mathbf{r},x,\mathbf{z})\pi^2(\mathbf{r},x,\mathbf{z})P(\mathbf{r},x,\mathbf{z}) \\
    &= \Sigma_{\mathbf{b,r},x,\mathbf{z}}\hat{\mu^3}(\mathbf{b,r},x_1,\mathbf{z})P(\mathbf{b}|\mathbf{r},x,\mathbf{z})\frac{P(\mathbf{r}\mid x_0,\mathbf{z})\mathds{1}[x=x_1]}
        {P(\mathbf{r}\mid x,\mathbf{z})P(x\mid \mathbf{z})}P(\mathbf{r},x,\mathbf{z})\\
    &=\Sigma_{\mathbf{b,r},\mathbf{z}}\hat{\mu^3}(\mathbf{b,r},x_1,\mathbf{z})P(\mathbf{b}|\mathbf{r},x_1,\mathbf{z})P(\mathbf{r}|x_0,\mathbf{z})p(\mathbf{z})\\
    &= \mathbb{E}[\pi^3(\mathbf{B,R},X,\mathbf{Z})\hat{\mu^3}(\mathbf{B,R},X,\mathbf{Z})]
\end{aligned}
\end{equation}

The other equality is 
\begin{equation}
        \mathbb{E}[\pi^2(\mathbf{R},X,\mathbf{Z})\hat{\mu^2}(\mathbf{R},X,\mathbf{Z})]=\mathbb{E}[\pi^1(X,\mathbf{Z})\hat{\mu^1_*}(X,\mathbf{Z})]
\end{equation}
where $\hat{\mu_*^1}\triangleq\mathbb{E}[\hat{\mu^2}(\mathbf{R},x_1,\mathbf{Z})|X,\mathbf{Z}]$ for any fixed $\hat{\mu^2}$.

This equality also holds since
\begin{equation}
\begin{aligned}
    &\mathbb{E}[\pi^1(X,\mathbf{Z})\hat{\mu^1_*}(X,\mathbf{Z})]=\mathbb{E}[\hat{\mu^1_*}(x_0,\,\mathbf{Z})]\\
    &= \mathbb{E}[\mathbb{E}[\hat{\mu^2}(\mathbf{R},x_1,\mathbf{Z})|x_0,\mathbf{Z}]] \\
    &=\Sigma_{\mathbf{z}}\Sigma_{\mathbf{r}}\hat{\mu^2}(\mathbf{r},x_1,\mathbf{z})P(\mathbf{r}|x_0,\mathbf{z})P(\mathbf{z}) \\
    &=  \mathbb{E}[\pi^2(\mathbf{R},X,\mathbf{Z})\hat{\mu^2}(\mathbf{R},X,\mathbf{Z})]
\end{aligned}
\end{equation}
\end{proof}

\twocolumn
\subsection{Grouping of High Cardinality Attributes}
\label{sec:app_grouping}
As outlined in \Cref{sec:causal_effect_estimation}, we group \textit{first\_name}, \textit{surname}, \textit{state}, and \textit{edu\_level} variables for the stable estimation of path-specific effects.
For \textit{first\_name}, we use statistics from SSA dataset, which provides distributions of names by gender and birth cohort.
Based on these statistics, we derive two grouping variables from each first name: a \textit{gender-typicality} category and an \textit{age-typicality} category.
Specifically, names whose probability exceeds a predefined threshold of 0.75 are labeled as \textit{male-typical} or \textit{female-typical}, while the remaining names are categorized as \textit{neutral}.
Age typicality is determined using the same procedure with a threshold of 0.5, producing three age groups (\textit{young}, \textit{mid}, and \textit{old}) with a \textit{neutral} category for low-confidence cases.
For surnames, we use the U.S. Census surname dataset, which reports the racial distribution associated with each surname.
Using these statistics, we assign surnames to \textit{Asian}, \textit{Black}, or \textit{White} categories when the corresponding probability exceeds a threshold of 0.5;
all others are labeled as \textit{neutral}.
Finally, we group the U.S. states into four major geographic regions (Northeast, Midwest, South, and West) following the U.S. Census Bureau regional classification~\cite{us_census_regions}. 
The \textit{edu\_level} variable from \cite{acs_pums} originally contains seven categories. We group them into three levels: 
\textit{low}, including \textit{High school graduate/GED}, \textit{Some college}; 
\textit{medium}, including \textit{Associate's degree}, \textit{Bachelor's}; 
and \textit{high}, including \textit{Master's}, \textit{Professional degree}, and \textit{Doctorate}.

\subsection{Hyperparameter Search}
\label{sec:hparam}
For the estimation of $P(X \mid \cdot)$, hyperparameters are selected by minimizing the Brier score, while for $\mathbb{E}[Y \mid \cdot]$, hyperparameters are selected by minimizing the mean squared error (MSE). The hyperparameter search space is defined as follows:

\begin{itemize}
    \item \texttt{n\_estimators}: $\{5, 10, 20, 50, 100\}$
    \item \texttt{max\_depth}: $\{1, 2, 3, 4, 5\}$
    \item \texttt{reg\_lambda}: $\{0.5, 1, 2, 5\}$
\end{itemize}

\section{Prompt for Resume Scoring}
We instruct the model to assign a suitability score to a resume given a job description using the following prompt. 
The model is instructed to return a single numerical score between 0 and 100 in a predefined JSON format. For VLM-based evaluation, the same prompt is used, except that the resume is provided as an image instead of text.

\begin{lstlisting}[basicstyle=\ttfamily,breaklines=true,breakindent=0pt]
You are a strict resume screening grader that grades resumes based on job descriptions.

Given a JOB DESCRIPTION and a RESUME, provide only a single score from 0 to 100.

Rules:
- Response format: {"score": XX.XX}
- Do NOT provide explanations or extra text.

JOB DESCRIPTION:
{job_text}

RESUME:
{resume_text}
\end{lstlisting}

\section{Additional Experimental Results}
\label{sec:full_additional_results}
\Cref{tab:pse_female} and \Cref{tab:pse_nonwhite} report the complete experimental results across all configurations and models for the protected attributes \textit{gender} and \textit{race}, respectively.
\onecolumn
{\scriptsize
\begin{table}[htbp]
\centering
\resizebox{\linewidth}{!}{%
\begin{tabular}{l c c c c c c}
\toprule
Model & Case & TE & NDE & NIE & RIE & BIE \\
\midrule
\multicolumn{7}{l}{\textbf{Accountants Auditors}} \\
\midrule
\multicolumn{7}{l}{\textbf{LLM}} \\
\rowcolor{cyan!15} Llama-3.1-8b-instruct & C1,C5 & \textbf{2.337 [2.162, 2.526]} & \textbf{3.365 [2.943, 3.777]} & \textbf{-1.028 [-1.452, -0.587]} & \textbf{-3.048 [-3.746, -2.384]} & \textbf{2.021 [1.374, 2.676]} \\
Mistral-7b-instruct-v0.2 & - & \textbf{2.388 [2.080, 2.697]} & 0.972 [-1.088, 3.044] & 1.416 [-0.540, 3.404] & 1.524 [-0.171, 3.321] & -0.108 [-0.927, 0.684] \\
\rowcolor{cyan!15} GPT-4o-mini & C1,C5 & \textbf{1.179 [1.053, 1.327]} & \textbf{2.217 [1.822, 2.589]} & \textbf{-1.038 [-1.393, -0.606]} & \textbf{-2.627 [-3.129, -2.118]} & \textbf{1.589 [1.174, 1.982]} \\
\rowcolor{cyan!15} Gemini-2.5-Flash-Lite (LLM) & C1,C5 & \textbf{1.076 [0.903, 1.264]} & \textbf{5.251 [4.762, 5.713]} & \textbf{-4.175 [-4.647, -3.693]} & \textbf{-6.207 [-6.844, -5.592]} & \textbf{2.032 [1.516, 2.552]} \\
\midrule
\multicolumn{7}{l}{\textbf{VLM (w/o profile image)}} \\
\rowcolor{cyan!15} Qwen2.5-vl-7b-instruct & C1,C5 & \textbf{0.255 [0.194, 0.317]} & \textbf{0.642 [0.381, 0.907]} & \textbf{-0.387 [-0.651, -0.134]} & \textbf{-1.462 [-1.728, -1.197]} & \textbf{1.074 [0.877, 1.262]} \\
\rowcolor{cyan!15} InternVL-8b & C5 & -0.022 [-0.047, 0.003] & -0.031 [-0.150, 0.086] & 0.009 [-0.109, 0.125] & \textbf{-0.231 [-0.338, -0.127]} & \textbf{0.240 [0.163, 0.312]} \\
\rowcolor{cyan!15} GPT-4o & C1,C5 & \textbf{0.288 [0.226, 0.357]} & \textbf{-1.150 [-1.575, -0.765]} & \textbf{1.438 [1.045, 1.860]} & \textbf{0.690 [0.244, 1.124]} & \textbf{0.748 [0.503, 0.990]} \\
Gemini-2.5-Flash-Lite (VLM) & - & \textbf{0.494 [0.404, 0.586]} & 0.111 [-0.394, 0.602] & 0.384 [-0.102, 0.917] & 0.019 [-0.414, 0.460] & \textbf{0.365 [0.099, 0.630]} \\
\midrule
\multicolumn{7}{l}{\textbf{VLM (w/ profile image)}} \\
\rowcolor{cyan!15} Qwen2.5-vl-7b-instruct & C5 & \textbf{0.513 [0.444, 0.582]} & 0.163 [-0.067, 0.390] & \textbf{0.350 [0.118, 0.570]} & \textbf{-0.773 [-1.057, -0.493]} & \textbf{1.123 [0.919, 1.313]} \\
\rowcolor{gray!15} InternVL-8b & C1 & \textbf{0.342 [0.317, 0.367]} & \textbf{0.139 [0.040, 0.255]} & \textbf{0.203 [0.082, 0.300]} & -0.039 [-0.148, 0.068] & \textbf{0.242 [0.177, 0.303]} \\
\rowcolor{cyan!15} GPT-4o & C1,C5 & \textbf{0.391 [0.325, 0.459]} & \textbf{0.385 [0.161, 0.593]} & 0.006 [-0.197, 0.218] & \textbf{-0.589 [-0.868, -0.306]} & \textbf{0.595 [0.425, 0.792]} \\
\rowcolor{cyan!15} Gemini-2.5-Flash-Lite (VLM) & C1,C5 & \textbf{0.841 [0.763, 0.928]} & \textbf{1.908 [1.501, 2.317]} & \textbf{-1.067 [-1.486, -0.669]} & \textbf{-1.673 [-2.044, -1.326]} & \textbf{0.607 [0.393, 0.805]} \\
\bottomrule
\end{tabular}%
}
\caption*{(Continued --- see full caption at \Cref{tab:pse_female})}
\end{table}

\begin{table}[htbp]
\centering
\resizebox{\linewidth}{!}{%
\begin{tabular}{l c c c c c c}
\toprule
Model & Case & TE & NDE & NIE & RIE & BIE \\
\midrule
\multicolumn{7}{l}{\textbf{Construction Laborers}} \\
\midrule
\multicolumn{7}{l}{\textbf{LLM}} \\
Llama-3.1-8b-instruct & - & -0.275 [-0.611, 0.053] & -0.134 [-0.615, 0.405] & -0.141 [-0.609, 0.281] & 0.238 [-0.316, 0.803] & -0.379 [-0.759, 0.006] \\
\rowcolor{blue!12} Mistral-7b-instruct-v0.2 & C1, C2 & 0.241 [-0.046, 0.575] & \textbf{-2.429 [-3.264, -1.532]} & \textbf{2.669 [1.845, 3.426]} & \textbf{2.428 [1.660, 3.298]} & 0.242 [-0.021, 0.480] \\
\rowcolor{cyan!15} GPT-4o-mini & C5 & \textbf{-1.231 [-1.443, -1.021]} & 0.040 [-0.335, 0.413] & \textbf{-1.271 [-1.595, -0.926]} & \textbf{-0.676 [-1.028, -0.333]} & \textbf{-0.595 [-0.833, -0.345]} \\
\rowcolor{gray!15} Gemini-2.5-Flash-Lite (LLM) & C1 & \textbf{-2.758 [-3.074, -2.456]} & \textbf{-1.358 [-1.863, -0.797]} & \textbf{-1.400 [-1.882, -0.991]} & -0.443 [-0.949, 0.003] & \textbf{-0.957 [-1.296, -0.606]} \\
\midrule
\multicolumn{7}{l}{\textbf{VLM (w/o profile image)}} \\
\rowcolor{cyan!15} Qwen2.5-vl-7b-instruct & C1,C5 & \textbf{-1.444 [-1.604, -1.281]} & \textbf{-0.316 [-0.607, -0.031]} & \textbf{-1.128 [-1.388, -0.873]} & \textbf{-0.576 [-0.842, -0.359]} & \textbf{-0.552 [-0.714, -0.387]} \\
\rowcolor{gray!15} InternVL-8b & C1 & \textbf{0.073 [0.048, 0.099]} & \textbf{0.140 [0.060, 0.213]} & -0.068 [-0.134, 0.002] & \textbf{-0.073 [-0.133, -0.009]} & 0.006 [-0.008, 0.019] \\
\rowcolor{cyan!15} GPT-4o & C5 & \textbf{-0.476 [-0.612, -0.345]} & 0.017 [-0.218, 0.232] & \textbf{-0.493 [-0.700, -0.284]} & \textbf{-0.349 [-0.560, -0.134]} & \textbf{-0.144 [-0.273, -0.022]} \\
\rowcolor{cyan!15} Gemini-2.5-Flash-Lite (VLM) & C1,C5 & \textbf{0.559 [0.339, 0.793]} & \textbf{2.956 [2.475, 3.472]} & \textbf{-2.397 [-2.833, -1.966]} & \textbf{-2.662 [-3.105, -2.235]} & \textbf{0.265 [0.150, 0.382]} \\
\midrule
\multicolumn{7}{l}{\textbf{VLM (w/ profile image)}} \\
\rowcolor{cyan!15} Qwen2.5-vl-7b-instruct & C1,C5 & \textbf{-1.061 [-1.212, -0.907]} & \textbf{-0.861 [-1.136, -0.588]} & -0.200 [-0.460, 0.044] & \textbf{0.396 [0.115, 0.648]} & \textbf{-0.596 [-0.792, -0.410]} \\
\rowcolor{cyan!15} InternVL-8b & C1,C5 & \textbf{7.164 [6.750, 7.542]} & \textbf{7.364 [6.938, 7.784]} & \textbf{-0.200 [-0.302, -0.101]} & \textbf{-0.389 [-0.501, -0.272]} & \textbf{0.189 [0.130, 0.254]} \\
\rowcolor{gray!15} GPT-4o & C1 & \textbf{0.598 [0.443, 0.755]} & \textbf{1.157 [0.863, 1.443]} & \textbf{-0.559 [-0.805, -0.286]} & \textbf{-0.424 [-0.677, -0.151]} & -0.135 [-0.261, 0.000] \\
\rowcolor{gray!15} Gemini-2.5-Flash-Lite (VLM) & C1 & \textbf{0.445 [0.215, 0.679]} & \textbf{0.436 [0.066, 0.835]} & 0.009 [-0.310, 0.357] & -0.204 [-0.503, 0.102] & \textbf{0.213 [0.118, 0.306]} \\
\bottomrule
\end{tabular}%
}
\caption*{(Continued --- see full caption at \Cref{tab:pse_female})}
\end{table}

\begin{table}[htbp]
\centering
\resizebox{\linewidth}{!}{%
\begin{tabular}{l c c c c c c}
\toprule
Model & Case & TE & NDE & NIE & RIE & BIE \\
\midrule
\multicolumn{7}{l}{\textbf{Elementary Middle School Teachers}} \\
\midrule
\multicolumn{7}{l}{\textbf{LLM}} \\
\rowcolor{green!12} Llama-3.1-8b-instruct & C3 & \textbf{0.975 [0.818, 1.121]} & 0.087 [-0.391, 0.583] & \textbf{0.888 [0.371, 1.349]} & 0.054 [-0.575, 0.674] & \textbf{0.834 [0.394, 1.279]} \\
\rowcolor{gray!15} Mistral-7b-instruct-v0.2 & C1 & \textbf{1.622 [1.277, 1.976]} & \textbf{5.457 [3.044, 7.795]} & \textbf{-3.835 [-6.111, -1.453]} & \textbf{-4.059 [-6.104, -1.721]} & 0.224 [-0.273, 0.642] \\
GPT-4o-mini & - & \textbf{-0.172 [-0.274, -0.066]} & -0.297 [-0.760, 0.152] & 0.125 [-0.329, 0.612] & -0.265 [-0.751, 0.241] & \textbf{0.390 [0.101, 0.692]} \\
\rowcolor{gray!15} Gemini-2.5-Flash-Lite (LLM) & C1 & \textbf{-0.297 [-0.467, -0.130]} & \textbf{-1.316 [-2.160, -0.579]} & \textbf{1.019 [0.269, 1.942]} & 0.436 [-0.396, 1.344] & \textbf{0.583 [0.128, 1.022]} \\
\midrule
\multicolumn{7}{l}{\textbf{VLM (w/o profile image)}} \\
\rowcolor{gray!15} Qwen2.5-vl-7b-instruct & C1 & \textbf{-0.274 [-0.332, -0.209]} & \textbf{0.738 [0.446, 1.003]} & \textbf{-1.012 [-1.280, -0.711]} & \textbf{-1.019 [-1.328, -0.692]} & 0.007 [-0.162, 0.190] \\
InternVL-8b & - & \textbf{-0.068 [-0.089, -0.048]} & -0.051 [-0.105, 0.002] & -0.017 [-0.069, 0.035] & -0.000 [-0.078, 0.074] & -0.017 [-0.065, 0.042] \\
\rowcolor{gray!15} GPT-4o & C1 & \textbf{-0.729 [-0.791, -0.666]} & \textbf{-0.355 [-0.648, -0.074]} & \textbf{-0.374 [-0.663, -0.073]} & \textbf{-0.474 [-0.710, -0.224]} & 0.101 [-0.095, 0.282] \\
\rowcolor{blue!12} Gemini-2.5-Flash-Lite (VLM) & C1, C2 & 0.240 [-0.071, 0.563] & \textbf{2.448 [0.351, 4.477]} & \textbf{-2.207 [-4.226, -0.166]} & \textbf{-2.137 [-4.248, -0.199]} & -0.071 [-0.445, 0.330] \\
\midrule
\multicolumn{7}{l}{\textbf{VLM (w/ profile image)}} \\
\rowcolor{cyan!15} Qwen2.5-vl-7b-instruct & C5 & \textbf{0.150 [0.097, 0.209]} & 0.194 [-0.026, 0.433] & -0.044 [-0.284, 0.185] & \textbf{-0.288 [-0.563, -0.007]} & \textbf{0.244 [0.088, 0.419]} \\
\rowcolor{gray!15} InternVL-8b & C1 & \textbf{0.115 [0.096, 0.135]} & \textbf{0.100 [0.008, 0.196]} & 0.015 [-0.077, 0.104] & 0.044 [-0.068, 0.142] & -0.029 [-0.080, 0.030] \\
\rowcolor{gray!15} GPT-4o & C1 & \textbf{-1.155 [-1.208, -1.098]} & \textbf{-0.905 [-1.209, -0.606]} & -0.250 [-0.559, 0.054] & \textbf{-0.340 [-0.635, -0.027]} & 0.090 [-0.003, 0.183] \\
Gemini-2.5-Flash-Lite (VLM) & - & \textbf{0.233 [0.158, 0.299]} & -0.130 [-0.531, 0.288] & 0.362 [-0.056, 0.766] & 0.281 [-0.160, 0.682] & \textbf{0.082 [0.023, 0.149]} \\
\bottomrule
\end{tabular}%
}
\caption*{(Continued --- see full caption at \Cref{tab:pse_female})}
\end{table}

\begin{table}[htbp]
\centering
\resizebox{\linewidth}{!}{%
\begin{tabular}{l c c c c c c}
\toprule
Model & Case & TE & NDE & NIE & RIE & BIE \\
\midrule
\multicolumn{7}{l}{\textbf{Registered Nurses}} \\
\midrule
\multicolumn{7}{l}{\textbf{LLM}} \\
\rowcolor{gray!15} Llama-3.1-8b-instruct & C1 & \textbf{0.477 [0.299, 0.701]} & \textbf{1.476 [0.791, 2.237]} & \textbf{-0.999 [-1.764, -0.274]} & -0.467 [-1.431, 0.482] & -0.532 [-1.197, 0.145] \\
\rowcolor{cyan!15} Mistral-7b-instruct-v0.2 & C1,C5 & \textbf{-1.645 [-1.952, -1.331]} & \textbf{3.420 [1.893, 4.726]} & \textbf{-5.064 [-6.351, -3.609]} & \textbf{-4.625 [-5.778, -3.234]} & \textbf{-0.440 [-0.835, -0.030]} \\
\rowcolor{gray!15} GPT-4o-mini & C1 & \textbf{0.230 [0.103, 0.370]} & \textbf{1.893 [1.255, 2.576]} & \textbf{-1.663 [-2.385, -0.985]} & \textbf{-1.249 [-1.961, -0.549]} & -0.414 [-0.934, 0.090] \\
\rowcolor{gray!15} Gemini-2.5-Flash-Lite (LLM) & C1 & \textbf{0.580 [0.410, 0.770]} & \textbf{1.246 [0.374, 2.236]} & -0.665 [-1.700, 0.197] & -0.467 [-1.478, 0.563] & -0.198 [-0.878, 0.413] \\
\midrule
\multicolumn{7}{l}{\textbf{VLM (w/o profile image)}} \\
\rowcolor{blue!12} Qwen2.5-vl-7b-instruct & C1, C2 & -0.040 [-0.098, 0.020] & \textbf{-0.369 [-0.658, -0.108]} & \textbf{0.329 [0.048, 0.621]} & 0.197 [-0.112, 0.533] & 0.131 [-0.071, 0.345] \\
\rowcolor{gray!15} InternVL-8b & C1 & \textbf{-0.171 [-0.195, -0.148]} & \textbf{0.282 [0.138, 0.426]} & \textbf{-0.453 [-0.604, -0.305]} & \textbf{-0.417 [-0.565, -0.255]} & -0.037 [-0.098, 0.030] \\
GPT-4o & - & \textbf{-0.205 [-0.269, -0.142]} & 0.043 [-0.321, 0.433] & -0.248 [-0.621, 0.118] & -0.268 [-0.694, 0.129] & 0.019 [-0.142, 0.179] \\
Gemini-2.5-Flash-Lite (VLM) & - & 0.110 [-0.307, 0.499] & -0.805 [-3.728, 2.117] & 0.915 [-2.080, 3.829] & 0.725 [-2.312, 3.760] & 0.190 [-0.402, 0.742] \\
\midrule
\multicolumn{7}{l}{\textbf{VLM (w/ profile image)}} \\
\rowcolor{orange!18} Qwen2.5-vl-7b-instruct & C4 & \textbf{0.288 [0.228, 0.344]} & -0.137 [-0.329, 0.046] & \textbf{0.424 [0.231, 0.614]} & \textbf{0.600 [0.340, 0.873]} & -0.176 [-0.371, 0.020] \\
\rowcolor{gray!15} InternVL-8b & C1 & \textbf{-0.567 [-0.596, -0.541]} & \textbf{0.198 [0.017, 0.392]} & \textbf{-0.765 [-0.956, -0.586]} & \textbf{-0.772 [-0.941, -0.606]} & 0.007 [-0.048, 0.055] \\
\rowcolor{gray!15} GPT-4o & C1 & \textbf{-0.306 [-0.363, -0.248]} & \textbf{-0.583 [-0.948, -0.193]} & 0.277 [-0.102, 0.647] & 0.191 [-0.204, 0.574] & 0.087 [-0.038, 0.214] \\
\rowcolor{gray!15} Gemini-2.5-Flash-Lite (VLM) & C1 & \textbf{0.446 [0.356, 0.523]} & \textbf{0.938 [0.269, 1.653]} & -0.492 [-1.196, 0.188] & -0.423 [-1.087, 0.237] & -0.069 [-0.267, 0.173] \\
\bottomrule
\end{tabular}%
}
\caption*{(Continued --- see full caption at \Cref{tab:pse_female})}
\end{table}

\begin{table}[htbp]
\centering
\resizebox{\linewidth}{!}{%
\begin{tabular}{l c c c c c c}
\toprule
Model & Case & TE & NDE & NIE & RIE & BIE \\
\midrule
\multicolumn{7}{l}{\textbf{Software Developers}} \\
\midrule
\multicolumn{7}{l}{\textbf{LLM}} \\
\rowcolor{cyan!15} Llama-3.1-8b-instruct & C5 & \textbf{1.192 [1.001, 1.392]} & -0.113 [-0.428, 0.209] & \textbf{1.305 [0.960, 1.624]} & \textbf{-1.456 [-2.069, -0.950]} & \textbf{2.761 [2.254, 3.269]} \\
\rowcolor{cyan!15} Mistral-7b-instruct-v0.2 & C1,C5 & \textbf{-0.226 [-0.363, -0.097]} & \textbf{-0.923 [-1.298, -0.600]} & \textbf{0.696 [0.407, 1.047]} & \textbf{-0.723 [-1.089, -0.331]} & \textbf{1.419 [1.163, 1.674]} \\
\rowcolor{red!12} GPT-4o-mini & C1, C2, C5 & -0.034 [-0.166, 0.108] & \textbf{-0.881 [-1.225, -0.547]} & \textbf{0.847 [0.524, 1.188]} & \textbf{-0.778 [-1.203, -0.319]} & \textbf{1.625 [1.283, 1.940]} \\
\rowcolor{cyan!15} Gemini-2.5-Flash-Lite (LLM) & C1,C5 & \textbf{0.655 [0.455, 0.853]} & \textbf{0.474 [0.007, 0.935]} & 0.181 [-0.289, 0.644] & \textbf{-2.529 [-3.216, -1.804]} & \textbf{2.710 [2.107, 3.275]} \\
\midrule
\multicolumn{7}{l}{\textbf{VLM (w/o profile image)}} \\
\rowcolor{cyan!15} Qwen2.5-vl-7b-instruct & C5 & \textbf{-0.292 [-0.372, -0.211]} & -0.151 [-0.325, 0.020] & -0.141 [-0.320, 0.045] & \textbf{-0.935 [-1.185, -0.681]} & \textbf{0.794 [0.589, 0.991]} \\
\rowcolor{cyan!15} InternVL-8b & C1,C5 & \textbf{-0.170 [-0.190, -0.149]} & \textbf{-0.220 [-0.281, -0.159]} & 0.050 [-0.010, 0.108] & \textbf{-0.069 [-0.132, -0.004]} & \textbf{0.119 [0.087, 0.150]} \\
\rowcolor{cyan!15} GPT-4o & C1, C5 & \textbf{-0.356 [-0.419, -0.295]} & \textbf{-0.632 [-0.762, -0.511]} & \textbf{0.276 [0.149, 0.400]} & \textbf{-0.326 [-0.519, -0.136]} & \textbf{0.603 [0.454, 0.750]} \\
Gemini-2.5-Flash-Lite (VLM) & - & 0.287 [-0.032, 0.581] & -0.376 [-1.791, 0.978] & 0.664 [-0.730, 2.067] & 0.604 [-0.715, 2.049] & 0.060 [-0.275, 0.354] \\
\midrule
\multicolumn{7}{l}{\textbf{VLM (w/ profile image)}} \\
\rowcolor{cyan!15} Qwen2.5-vl-7b-instruct & C1,C5 & \textbf{-0.094 [-0.175, -0.024]} & \textbf{-0.321 [-0.492, -0.153]} & \textbf{0.228 [0.053, 0.404]} & \textbf{-0.610 [-0.830, -0.358]} & \textbf{0.838 [0.651, 1.024]} \\
\rowcolor{cyan!15} InternVL-8b & C1,C5 & \textbf{0.327 [0.304, 0.350]} & \textbf{0.399 [0.330, 0.466]} & \textbf{-0.072 [-0.142, -0.003]} & \textbf{-0.324 [-0.414, -0.239]} & \textbf{0.251 [0.203, 0.298]} \\
\rowcolor{gray!15} GPT-4o & C1 & \textbf{-0.193 [-0.278, -0.121]} & \textbf{-0.530 [-0.780, -0.278]} & \textbf{0.337 [0.105, 0.597]} & -0.174 [-0.428, 0.107] & \textbf{0.511 [0.367, 0.646]} \\
\rowcolor{cyan!15} Gemini-2.5-Flash-Lite (VLM) & C1,C5 & \textbf{0.497 [0.419, 0.574]} & \textbf{0.689 [0.381, 0.980]} & -0.192 [-0.477, 0.102] & \textbf{-0.355 [-0.624, -0.055]} & \textbf{0.163 [0.120, 0.215]} \\
\bottomrule
\end{tabular}%
}
\caption{Path-specific effects when $x_0 = \text{Female}$. Bold values indicate confidence intervals excluding zero. The Case column lists matching cases among C1--C5. Rows corresponding only to C1 are shaded gray, while rows containing C2, C3, C4, or C5 use distinct highlight colors.}
\label{tab:pse_female}
\end{table}
}
{\scriptsize
\begin{table}[htbp]
\centering
\resizebox{\linewidth}{!}{%
\begin{tabular}{l c c c c c c}
\toprule
Model & Case & TE & NDE & NIE & RIE & BIE \\
\midrule
\multicolumn{7}{l}{\textbf{Accountants Auditors}} \\
\midrule
\multicolumn{7}{l}{\textbf{LLM}} \\
\rowcolor{cyan!15} Llama-3.1-8b-instruct & C1,C5 & \textbf{0.658 [0.415, 0.904]} & \textbf{-0.417 [-0.743, -0.040]} & \textbf{1.075 [0.708, 1.453]} & \textbf{-1.349 [-1.802, -0.913]} & \textbf{2.424 [2.109, 2.750]} \\
\rowcolor{gray!15} Mistral-7b-instruct-v0.2 & C1 & \textbf{-4.267 [-4.637, -3.946]} & \textbf{-2.870 [-3.778, -1.972]} & \textbf{-1.397 [-2.281, -0.578]} & \textbf{-1.409 [-2.209, -0.541]} & 0.012 [-0.289, 0.304] \\
\rowcolor{cyan!15} GPT-4o-mini & C5 & 0.124 [-0.061, 0.293] & 0.024 [-0.285, 0.303] & 0.100 [-0.189, 0.408] & \textbf{-1.117 [-1.430, -0.801]} & \textbf{1.217 [1.002, 1.421]} \\
\rowcolor{cyan!15} Gemini-2.5-Flash-Lite (LLM) & C1,C5 & \textbf{1.173 [0.916, 1.392]} & \textbf{0.687 [0.393, 0.971]} & \textbf{0.486 [0.201, 0.819]} & \textbf{-1.382 [-1.711, -1.008]} & \textbf{1.868 [1.610, 2.131]} \\
\midrule
\multicolumn{7}{l}{\textbf{VLM (w/o profile image)}} \\
\rowcolor{cyan!15} Qwen2.5-vl-7b-instruct & C5 & \textbf{0.202 [0.124, 0.281]} & 0.045 [-0.060, 0.156] & \textbf{0.157 [0.033, 0.264]} & \textbf{-0.325 [-0.466, -0.200]} & \textbf{0.482 [0.392, 0.575]} \\
\rowcolor{cyan!15} InternVL-8b & C1,C5 & \textbf{0.117 [0.087, 0.147]} & \textbf{0.125 [0.077, 0.175]} & -0.008 [-0.053, 0.043] & \textbf{-0.057 [-0.111, -0.008]} & \textbf{0.049 [0.022, 0.077]} \\
\rowcolor{cyan!15} GPT-4o & C1,C5 & \textbf{0.311 [0.230, 0.401]} & \textbf{0.166 [0.049, 0.283]} & \textbf{0.145 [0.009, 0.276]} & \textbf{-0.619 [-0.761, -0.469]} & \textbf{0.764 [0.668, 0.861]} \\
\rowcolor{green!12} Gemini-2.5-Flash-Lite (VLM) & C3 & \textbf{0.365 [0.256, 0.475]} & 0.033 [-0.181, 0.246] & \textbf{0.332 [0.139, 0.514]} & 0.060 [-0.131, 0.258] & \textbf{0.272 [0.215, 0.332]} \\
\midrule
\multicolumn{7}{l}{\textbf{VLM (w/ profile image)}} \\
\rowcolor{cyan!15} Qwen2.5-vl-7b-instruct & C1,C5 & -0.074 [-0.146, 0.004] & \textbf{-0.102 [-0.202, -0.011]} & 0.028 [-0.076, 0.126] & \textbf{-0.495 [-0.612, -0.383]} & \textbf{0.524 [0.433, 0.616]} \\
\rowcolor{gray!15} InternVL-8b & C1 & \textbf{0.108 [0.081, 0.135]} & \textbf{0.090 [0.033, 0.149]} & 0.018 [-0.036, 0.071] & -0.010 [-0.069, 0.042] & \textbf{0.028 [0.003, 0.053]} \\
\rowcolor{cyan!15} GPT-4o & C5 & \textbf{-0.297 [-0.371, -0.219]} & 0.025 [-0.094, 0.150] & \textbf{-0.322 [-0.446, -0.198]} & \textbf{-1.006 [-1.148, -0.879]} & \textbf{0.685 [0.586, 0.786]} \\
\rowcolor{green!12} Gemini-2.5-Flash-Lite (VLM) & C3 & \textbf{0.351 [0.246, 0.447]} & 0.119 [-0.078, 0.321] & \textbf{0.232 [0.047, 0.420]} & -0.075 [-0.265, 0.122] & \textbf{0.307 [0.240, 0.367]} \\
\bottomrule
\end{tabular}%
}
\caption*{(Continued --- see full caption at \Cref{tab:pse_nonwhite})}
\end{table}

\begin{table}[htbp]
\centering
\resizebox{\linewidth}{!}{%
\begin{tabular}{l c c c c c c}
\toprule
Model & Case & TE & NDE & NIE & RIE & BIE \\
\midrule
\multicolumn{7}{l}{\textbf{Construction Laborers}} \\
\midrule
\multicolumn{7}{l}{\textbf{LLM}} \\
\rowcolor{gray!15} Llama-3.1-8b-instruct & C1 & \textbf{0.815 [0.516, 1.074]} & \textbf{0.364 [0.135, 0.620]} & \textbf{0.451 [0.278, 0.609]} & -0.062 [-0.208, 0.077] & \textbf{0.513 [0.372, 0.648]} \\
\rowcolor{cyan!15} Mistral-7b-instruct-v0.2 & C5 & -0.149 [-0.328, 0.037] & 0.055 [-0.133, 0.240] & \textbf{-0.204 [-0.281, -0.121]} & \textbf{-0.325 [-0.406, -0.237]} & \textbf{0.121 [0.075, 0.166]} \\
\rowcolor{green!12} GPT-4o-mini & C3 & \textbf{0.350 [0.184, 0.529]} & 0.071 [-0.076, 0.215] & \textbf{0.280 [0.161, 0.393]} & -0.006 [-0.115, 0.099] & \textbf{0.285 [0.188, 0.381]} \\
\rowcolor{cyan!15} Gemini-2.5-Flash-Lite (LLM) & C5 & \textbf{0.473 [0.224, 0.703]} & -0.200 [-0.434, 0.022] & \textbf{0.673 [0.509, 0.840]} & \textbf{0.312 [0.144, 0.476]} & \textbf{0.361 [0.200, 0.506]} \\
\midrule
\multicolumn{7}{l}{\textbf{VLM (w/o profile image)}} \\
\rowcolor{blue!12} Qwen2.5-vl-7b-instruct & C1, C2 & 0.022 [-0.069, 0.119] & \textbf{-0.236 [-0.324, -0.139]} & \textbf{0.257 [0.189, 0.322]} & 0.023 [-0.044, 0.085] & \textbf{0.234 [0.180, 0.291]} \\
\rowcolor{cyan!15} InternVL-8b & C1,C5 & \textbf{0.045 [0.026, 0.064]} & \textbf{0.032 [0.012, 0.051]} & \textbf{0.013 [0.009, 0.018]} & \textbf{0.007 [0.002, 0.012]} & \textbf{0.006 [0.004, 0.008]} \\
\rowcolor{cyan!15} GPT-4o & C5 & \textbf{0.223 [0.135, 0.318]} & -0.009 [-0.106, 0.090] & \textbf{0.232 [0.166, 0.302]} & \textbf{0.091 [0.029, 0.154]} & \textbf{0.141 [0.088, 0.191]} \\
Gemini-2.5-Flash-Lite (VLM) & - & 0.154 [-0.441, 0.765] & -0.017 [-0.704, 0.617] & 0.171 [-0.083, 0.429] & 0.151 [-0.129, 0.399] & 0.020 [-0.115, 0.162] \\
\midrule
\multicolumn{7}{l}{\textbf{VLM (w/ profile image)}} \\
\rowcolor{cyan!15} Qwen2.5-vl-7b-instruct & C1,C5 & \textbf{0.209 [0.105, 0.312]} & \textbf{-0.158 [-0.248, -0.071]} & \textbf{0.368 [0.299, 0.439]} & \textbf{0.113 [0.043, 0.180]} & \textbf{0.255 [0.191, 0.318]} \\
\rowcolor{gray!15} InternVL-8b & C1 & \textbf{-0.263 [-0.345, -0.172]} & \textbf{-0.409 [-0.516, -0.309]} & \textbf{0.145 [0.083, 0.212]} & 0.044 [-0.010, 0.098] & \textbf{0.101 [0.060, 0.148]} \\
\rowcolor{cyan!15} GPT-4o & C1,C5 & \textbf{-0.385 [-0.484, -0.274]} & \textbf{-0.743 [-0.849, -0.634]} & \textbf{0.358 [0.277, 0.436]} & \textbf{0.183 [0.113, 0.254]} & \textbf{0.176 [0.124, 0.223]} \\
\rowcolor{cyan!15} Gemini-2.5-Flash-Lite (VLM) & C1,C5 & \textbf{-0.151 [-0.311, -0.011]} & \textbf{-0.346 [-0.510, -0.203]} & \textbf{0.194 [0.127, 0.266]} & \textbf{0.255 [0.183, 0.332]} & \textbf{-0.060 [-0.093, -0.024]} \\
\bottomrule
\end{tabular}%
}
\caption*{(Continued --- see full caption at \Cref{tab:pse_nonwhite})}
\end{table}

\begin{table}[htbp]
\centering
\resizebox{\linewidth}{!}{%
\begin{tabular}{l c c c c c c}
\toprule
Model & Case & TE & NDE & NIE & RIE & BIE \\
\midrule
\multicolumn{7}{l}{\textbf{Elementary Middle School Teachers}} \\
\midrule
\multicolumn{7}{l}{\textbf{LLM}} \\
\rowcolor{cyan!15} Llama-3.1-8b-instruct & C1,C5 & \textbf{0.382 [0.170, 0.563]} & \textbf{0.266 [0.145, 0.390]} & 0.117 [-0.066, 0.306] & \textbf{-0.162 [-0.323, -0.004]} & \textbf{0.279 [0.095, 0.477]} \\
\rowcolor{cyan!15} Mistral-7b-instruct-v0.2 & C1,C5 & \textbf{-2.813 [-3.242, -2.347]} & \textbf{-2.086 [-2.612, -1.541]} & \textbf{-0.727 [-1.066, -0.345]} & \textbf{-1.006 [-1.362, -0.604]} & \textbf{0.279 [0.113, 0.439]} \\
\rowcolor{gray!15} GPT-4o-mini & C1 & \textbf{0.233 [0.101, 0.367]} & \textbf{0.124 [0.022, 0.229]} & 0.109 [-0.028, 0.249] & -0.073 [-0.194, 0.053] & \textbf{0.182 [0.053, 0.305]} \\
\rowcolor{cyan!15} Gemini-2.5-Flash-Lite (LLM) & C1,C5 & 0.179 [-0.026, 0.391] & \textbf{0.275 [0.095, 0.436]} & -0.096 [-0.319, 0.103] & \textbf{-0.292 [-0.482, -0.092]} & \textbf{0.196 [0.005, 0.392]} \\
\midrule
\multicolumn{7}{l}{\textbf{VLM (w/o profile image)}} \\
\rowcolor{gray!15} Qwen2.5-vl-7b-instruct & C1 & \textbf{0.131 [0.056, 0.211]} & \textbf{0.066 [0.016, 0.127]} & 0.065 [-0.003, 0.135] & -0.032 [-0.101, 0.034] & \textbf{0.097 [0.027, 0.164]} \\
\rowcolor{gray!15} InternVL-8b & C1 & \textbf{-0.055 [-0.079, -0.030]} & \textbf{-0.085 [-0.108, -0.063]} & \textbf{0.030 [0.007, 0.054]} & -0.008 [-0.031, 0.017] & \textbf{0.038 [0.016, 0.060]} \\
\rowcolor{cyan!15} GPT-4o & C1,C5 & 0.063 [-0.005, 0.129] & \textbf{0.083 [0.016, 0.158]} & -0.020 [-0.086, 0.042] & \textbf{-0.110 [-0.173, -0.047]} & \textbf{0.090 [0.034, 0.146]} \\
Gemini-2.5-Flash-Lite (VLM) & - & -0.222 [-0.576, 0.137] & -0.360 [-0.814, 0.121] & 0.138 [-0.201, 0.496] & 0.097 [-0.263, 0.449] & 0.041 [-0.115, 0.202] \\
\midrule
\multicolumn{7}{l}{\textbf{VLM (w/ profile image)}} \\
\rowcolor{cyan!15} Qwen2.5-vl-7b-instruct & C1,C5 & -0.051 [-0.118, 0.018] & \textbf{-0.069 [-0.118, -0.017]} & 0.018 [-0.050, 0.092] & \textbf{-0.104 [-0.164, -0.045]} & \textbf{0.122 [0.052, 0.190]} \\
\rowcolor{cyan!15} InternVL-8b & C1,C5 & \textbf{-0.331 [-0.357, -0.306]} & \textbf{-0.415 [-0.442, -0.390]} & \textbf{0.084 [0.058, 0.109]} & \textbf{0.041 [0.016, 0.066]} & \textbf{0.043 [0.021, 0.063]} \\
\rowcolor{cyan!15} GPT-4o & C1,C5 & \textbf{-0.500 [-0.566, -0.435]} & \textbf{-0.402 [-0.467, -0.338]} & \textbf{-0.098 [-0.156, -0.042]} & \textbf{-0.179 [-0.236, -0.127]} & \textbf{0.081 [0.034, 0.130]} \\
\rowcolor{orange!18} Gemini-2.5-Flash-Lite (VLM) & C4 & \textbf{0.105 [0.019, 0.192]} & 0.019 [-0.098, 0.142] & \textbf{0.087 [0.008, 0.166]} & \textbf{0.095 [0.009, 0.175]} & -0.008 [-0.038, 0.020] \\
\bottomrule
\end{tabular}%
}
\caption*{(Continued --- see full caption at \Cref{tab:pse_nonwhite})}
\end{table}

\begin{table}[htbp]
\centering
\resizebox{\linewidth}{!}{%
\begin{tabular}{l c c c c c c}
\toprule
Model & Case & TE & NDE & NIE & RIE & BIE \\
\midrule
\multicolumn{7}{l}{\textbf{Registered Nurses}} \\
\midrule
\multicolumn{7}{l}{\textbf{LLM}} \\
\rowcolor{blue!12} Llama-3.1-8b-instruct & C1, C2 & 0.081 [-0.071, 0.236] & \textbf{-0.387 [-0.548, -0.210]} & \textbf{0.468 [0.270, 0.672]} & \textbf{0.352 [0.112, 0.597]} & 0.116 [-0.106, 0.338] \\
\rowcolor{cyan!15} Mistral-7b-instruct-v0.2 & C1,C5 & \textbf{-3.120 [-3.362, -2.896]} & \textbf{-1.844 [-2.289, -1.410]} & \textbf{-1.276 [-1.727, -0.879]} & \textbf{-1.049 [-1.538, -0.623]} & \textbf{-0.226 [-0.330, -0.127]} \\
\rowcolor{gray!15} GPT-4o-mini & C1 & \textbf{-0.187 [-0.294, -0.072]} & \textbf{-0.309 [-0.435, -0.189]} & 0.122 [-0.023, 0.258] & 0.099 [-0.062, 0.265] & 0.023 [-0.113, 0.168] \\
\rowcolor{orange!18} Gemini-2.5-Flash-Lite (LLM) & C4 & \textbf{0.195 [0.045, 0.340]} & -0.082 [-0.269, 0.112] & \textbf{0.276 [0.078, 0.481]} & \textbf{0.256 [0.016, 0.481]} & 0.020 [-0.175, 0.215] \\
\midrule
\multicolumn{7}{l}{\textbf{VLM (w/o profile image)}} \\
\rowcolor{gray!15} Qwen2.5-vl-7b-instruct & C1 & \textbf{0.068 [0.016, 0.123]} & \textbf{0.063 [0.000, 0.125]} & 0.004 [-0.060, 0.076] & 0.010 [-0.060, 0.079] & -0.006 [-0.066, 0.056] \\
\rowcolor{green!12} InternVL-8b & C3 & \textbf{0.074 [0.056, 0.093]} & 0.025 [-0.003, 0.054] & \textbf{0.049 [0.022, 0.075]} & -0.005 [-0.035, 0.023] & \textbf{0.054 [0.036, 0.074]} \\
\rowcolor{green!12} GPT-4o & C3 & \textbf{-0.073 [-0.121, -0.020]} & 0.009 [-0.074, 0.104] & \textbf{-0.082 [-0.169, -0.000]} & 0.011 [-0.081, 0.095] & \textbf{-0.093 [-0.158, -0.031]} \\
Gemini-2.5-Flash-Lite (VLM) & - & -0.120 [-0.408, 0.182] & -0.252 [-0.797, 0.289] & 0.132 [-0.297, 0.602] & 0.090 [-0.366, 0.551] & 0.042 [-0.168, 0.256] \\
\midrule
\multicolumn{7}{l}{\textbf{VLM (w/ profile image)}} \\
Qwen2.5-vl-7b-instruct & - & -0.028 [-0.074, 0.022] & -0.032 [-0.075, 0.015] & 0.004 [-0.051, 0.059] & -0.066 [-0.132, 0.005] & \textbf{0.070 [0.011, 0.132]} \\
\rowcolor{cyan!15} InternVL-8b & C1,C5 & \textbf{-0.218 [-0.238, -0.196]} & \textbf{-0.303 [-0.343, -0.266]} & \textbf{0.085 [0.052, 0.119]} & \textbf{0.049 [0.017, 0.081]} & \textbf{0.037 [0.022, 0.050]} \\
\rowcolor{gray!15} GPT-4o & C1 & \textbf{-0.516 [-0.569, -0.467]} & \textbf{-0.430 [-0.504, -0.365]} & \textbf{-0.086 [-0.156, -0.018]} & -0.008 [-0.077, 0.064] & \textbf{-0.078 [-0.124, -0.029]} \\
\rowcolor{orange!18} Gemini-2.5-Flash-Lite (VLM) & C4 & \textbf{0.217 [0.139, 0.294]} & 0.066 [-0.068, 0.186] & \textbf{0.151 [0.054, 0.254]} & \textbf{0.123 [0.027, 0.224]} & 0.029 [-0.014, 0.073] \\
\bottomrule
\end{tabular}%
}
\caption*{(Continued --- see full caption at \Cref{tab:pse_nonwhite})}
\end{table}

\begin{table}[htbp]
\centering
\resizebox{\linewidth}{!}{%
\begin{tabular}{l c c c c c c}
\toprule
Model & Case & TE & NDE & NIE & RIE & BIE \\
\midrule
\multicolumn{7}{l}{\textbf{Software Developers}} \\
\midrule
\multicolumn{7}{l}{\textbf{LLM}} \\
\rowcolor{gray!15} Llama-3.1-8b-instruct & C1 & \textbf{-0.412 [-0.553, -0.268]} & \textbf{0.552 [0.285, 0.783]} & \textbf{-0.964 [-1.226, -0.688]} & -0.261 [-0.619, 0.070] & \textbf{-0.703 [-0.994, -0.413]} \\
\rowcolor{gray!15} Mistral-7b-instruct-v0.2 & C1 & \textbf{-0.825 [-0.937, -0.722]} & \textbf{-0.656 [-1.432, -0.015]} & -0.169 [-0.826, 0.583] & -0.423 [-0.887, 0.061] & 0.254 [-0.021, 0.558] \\
\rowcolor{gray!15} GPT-4o-mini & C1 & \textbf{-0.527 [-0.647, -0.411]} & \textbf{-0.401 [-0.679, -0.134]} & -0.127 [-0.389, 0.150] & -0.145 [-0.459, 0.141] & 0.018 [-0.175, 0.218] \\
\rowcolor{orange!18} Gemini-2.5-Flash-Lite (LLM) & C4 & \textbf{-0.767 [-0.930, -0.615]} & -0.064 [-0.339, 0.232] & \textbf{-0.703 [-1.002, -0.392]} & \textbf{-0.753 [-1.108, -0.403]} & 0.050 [-0.237, 0.342] \\
\midrule
\multicolumn{7}{l}{\textbf{VLM (w/o profile image)}} \\
\rowcolor{cyan!15} Qwen2.5-vl-7b-instruct & C5 & \textbf{-0.795 [-0.852, -0.737]} & 0.099 [-0.032, 0.223] & \textbf{-0.894 [-1.027, -0.763]} & \textbf{-0.302 [-0.436, -0.166]} & \textbf{-0.593 [-0.674, -0.509]} \\
\rowcolor{gray!15} InternVL-8b & C1 & \textbf{-0.085 [-0.098, -0.072]} & \textbf{0.061 [0.026, 0.100]} & \textbf{-0.146 [-0.186, -0.109]} & -0.019 [-0.056, 0.021] & \textbf{-0.128 [-0.143, -0.111]} \\
\rowcolor{orange!18} GPT-4o & C4 & \textbf{-0.225 [-0.282, -0.176]} & 0.040 [-0.057, 0.138] & \textbf{-0.265 [-0.365, -0.157]} & \textbf{-0.194 [-0.309, -0.075]} & -0.071 [-0.142, 0.004] \\
Gemini-2.5-Flash-Lite (VLM) & - & 0.059 [-0.164, 0.312] & 0.672 [-0.233, 1.640] & -0.613 [-1.548, 0.284] & -0.591 [-1.498, 0.305] & -0.022 [-0.213, 0.166] \\
\midrule
\multicolumn{7}{l}{\textbf{VLM (w/ profile image)}} \\
\rowcolor{gray!15} Qwen2.5-vl-7b-instruct & C1 & \textbf{-1.077 [-1.128, -1.027]} & \textbf{-0.504 [-0.629, -0.380]} & \textbf{-0.574 [-0.705, -0.440]} & 0.062 [-0.089, 0.193] & \textbf{-0.635 [-0.711, -0.560]} \\
\rowcolor{gray!15} InternVL-8b & C1 & \textbf{-0.262 [-0.280, -0.246]} & \textbf{-0.165 [-0.204, -0.124]} & \textbf{-0.096 [-0.137, -0.060]} & 0.003 [-0.041, 0.046] & \textbf{-0.100 [-0.120, -0.079]} \\
\rowcolor{gray!15} GPT-4o & C1 & \textbf{-0.627 [-0.687, -0.564]} & \textbf{-0.918 [-1.246, -0.610]} & 0.290 [-0.034, 0.622] & 0.252 [-0.011, 0.493] & 0.039 [-0.049, 0.137] \\
\rowcolor{cyan!15} Gemini-2.5-Flash-Lite (VLM) & C1,C5 & \textbf{-0.120 [-0.174, -0.066]} & \textbf{1.230 [1.022, 1.459]} & \textbf{-1.350 [-1.568, -1.144]} & \textbf{-1.178 [-1.384, -0.980]} & \textbf{-0.171 [-0.210, -0.128]} \\
\bottomrule
\end{tabular}%
}
\caption{Path-specific effects for race experiments ($x_0 = \text{Non-White}$). Bold values indicate confidence intervals excluding zero. The Case column lists matching cases among C1--C5. Rows corresponding only to C1 are shaded gray, while rows containing C2, C3, C4, or C5 use distinct highlight colors.}
\label{tab:pse_nonwhite}
\end{table}
}
\twocolumn

\end{document}